%% file: main.tex
\documentclass[twocolumn,secnumarabic,nobibnotes,superscriptaddress,aps,longbibliography]{revtex4-2}

\usepackage{color, xcolor, colortbl}
\usepackage{graphicx}
\usepackage{amsmath,amssymb,amsthm,amsfonts}
\usepackage{algorithm}
\usepackage{algorithmic}
\usepackage{dcolumn}
\usepackage{epstopdf}
\usepackage{bm}
\usepackage[caption=false]{subfig}
\usepackage{appendix}
\usepackage{multirow}
\usepackage{braket}
\usepackage[english]{babel}
\usepackage[pdfencoding=auto,breaklinks=true]{hyperref}
\usepackage[capitalize]{cleveref}
\usepackage[T1]{fontenc}
\usepackage{xpatch}
\usepackage{tikz}
\usepackage{adjustbox}
\usepackage{xspace}
\usepackage[roman]{complexity}

\newcommand{\bvec}[1]{\mathbf{#1}}

\renewcommand{\vr}{\bvec{r}}

\newcommand{\mc}[1]{\mathcal{#1}}

\newcommand{\abs}[1]{\left\lvert#1\right\rvert}
\newcommand{\norm}[1]{\left\lVert#1\right\rVert}

\newcommand{\ud}{\,\mathrm{d}}
\newcommand{\Or}{\mathcal{O}}

\newcommand{\RR}{\mathbb{R}}
\newcommand{\CC}{\mathbb{C}}

\newtheorem{thm}{\protect\theoremname}
\theoremstyle{plain}

\theoremstyle{plain}
\newtheorem{lem}[thm]{\protect\lemmaname}
\theoremstyle{plain}

\theoremstyle{plain}
\newtheorem*{lem*}{\protect\lemmaname}
\theoremstyle{plain}

\theoremstyle{plain}

\providecommand{\definitionname}{Definition}
\providecommand{\assumptionname}{Assumption}
\providecommand{\corollaryname}{Corollary}
\providecommand{\lemmaname}{Lemma}
\providecommand{\propositionname}{Proposition}
\providecommand{\remarkname}{Remark}
\providecommand{\theoremname}{Theorem}
\providecommand{\conjecturename}{Conjecture}

\newcommand{\REV}[1]{{#1}}

\newcommand{\prlsection}[1]{\paragraph*{#1---}}
\usetikzlibrary{quantikz}

\usetikzlibrary{fit}
\tikzset{%
  highlight/.style={rectangle,rounded corners,fill=blue!15,draw,fill opacity=0.3,thick,inner sep=0pt}
}

\usepackage{amsbsy}
\newcommand{\QUICS}{Joint Center for Quantum Information and Computer Science, University of Maryland, MD 20742, USA}
\newcommand{\DeptMath}{Department of Mathematics, University of California, Berkeley, CA 94720, USA}
\newcommand{\LBLMath}{Applied Mathematics and Computational Research Division, Lawrence Berkeley National Laboratory, Berkeley, CA 94720, USA}
\newcommand{\CIQC}{Challenge Institute for Quantum Computation, University of California, Berkeley, CA 94720, USA}
\newcommand{\Simons}{Simons Institute for the Theory of Computing, University of California, Berkeley, CA 94720, USA}
\newcommand{\CTP}{Center for Theoretical Physics, Massachusetts Institute of Technology, Cambridge, MA 02139, USA}

\begin{document}

\title{Linear combination of Hamiltonian simulation for nonunitary dynamics with optimal state preparation cost}
\author{Dong An}
\affiliation{\QUICS}
\author{Jin-Peng Liu}
\affiliation{\DeptMath}
\affiliation{\Simons}
\affiliation{\CTP}
\author{Lin Lin}
\email{linlin@math.berkeley.edu}
\affiliation{\DeptMath}
\affiliation{\LBLMath}
\affiliation{\CIQC}

\date{Latest revision: \today}

\begin{abstract}
We propose a simple method for simulating a general class of non-unitary dynamics as a linear combination of Hamiltonian simulation (LCHS) problems. LCHS does \emph{not} rely on converting the problem into a dilated linear system problem, or on the spectral mapping theorem. The latter is the mathematical foundation of many quantum algorithms for solving a wide variety of tasks involving non-unitary processes, such as the quantum singular value transformation. The LCHS method can achieve optimal cost in terms of state preparation. We also demonstrate an application for open quantum dynamics simulation using the complex absorbing potential method with near-optimal dependence on all parameters.

\end{abstract}
\maketitle

\prlsection{Introduction}
Fault-tolerant quantum computers are expected to excel in simulating unitary dynamics, such as the dynamics of a quantum state under a Hamiltonian. Most applications in scientific and engineering computations involve non-unitary dynamics and processes. Therefore, efficient quantum algorithms are the key for unlocking the full potential of quantum computers to achieve comparable speedup in these general tasks. Quantum phase estimation (QPE)~\cite{NielsenChuang2000,KitaevShenVyalyi2002} is the first algorithm to bridge this gap. QPE  stores the eigenvalues of the Hamiltonian in an ancilla quantum register, and can be used to solve a wide range of problems, including amplitude estimation~\cite{BrassardHoyerMoscaEtAl2002}, linear systems~\cite{HarrowHassidimLloyd2009}, and differential equations~\cite{Berry2014,ChildsLiu2020}.
There have been two significant improvements over QPE based methods. The first is  linear combinations of unitaries (LCU)~\cite{ChildsWiebe2012}, which provides  an exponential improvement in precision for tasks such as the Hamiltonian simulations~\cite{BerryChildsCleveEtAl2015,BerryChildsKothari2015} and linear systems~\cite{ChildsKothariSomma2017}. 
Here the exponential improvement refers to the cost of preparing a quantum state in a register, which also directly translates into a polynomial improvement in precision  when the final  measurement cost is taken into account.
The second is quantum signal processing (QSP)~\cite{LowChuang2017}, and its generalization including quantum singular value transformation (QSVT)~\cite{GilyenSuLowEtAl2019}, and quantum eigenvalue transformation of unitary matrices (QETU)~\cite{DongLinTong2022}. Compared to LCU, QSP based approaches can achieve a similar level of accuracy but with a much more compact quantum circuit and a minimal number of ancilla qubits. 

The unifying mathematical argument that underlies many of these approaches is the spectral mapping theorem for Hermitian matrices: Let $A$ be a Hermitian matrix and $f(A)$ be a real-valued matrix function defined on the eigenvalues of $A$. Then, the eigenvalues of the matrix $f(A)$ are equal to the values of the classical function $f$ applied to the eigenvalues of $A$, i.e., if $\lambda_1, \dots, \lambda_N$ are the eigenvalues of $A$, then the the eigenvalues of $f(A)$ are $f(\lambda_1), \dots, f(\lambda_N)$. Furthermore, $f(A)$ is diagonalized by the same unitary matrix that diagonalizes $A$.
For instance,  the quantum linear system problem corresponds to $f(A)=A^{-1}$, Hamiltonian simulation corresponds to $f(A)=\cos(At), \sin(At)$ (for the real and imaginary parts of $e^{-iAt}$), and  Gibbs state preparation corresponds to $f(A)=e^{-A}$. 

The limitation of this matrix-function-based perspective can be readily observed when solving a general differential equation 
\begin{equation}
\partial_t u(t) = -A(t) u(t)+b(t), \quad u(0)=u_0.
\label{eqn:inhom_general_diff_eq}
\end{equation}
When $b(t)=0$ and $A(t)=A\in\CC^{N\times N}$ is a general time-independent matrix, the system has a closed form solution $u(t)=e^{-At}u_0$. Even in this case, the eigenvalues of $A$ may not be real, the eigenvectors of $A$ may not form a unitary matrix, or $A$ may not be diagonalizable at all. These difficulties prevent us from applying techniques such as QSVT to implement $f(A)=e^{-At}$ on a quantum computer, and the situation becomes much more complicated when $A,b$ are not some fixed matrices and vectors, but are time-dependent. To solve \cref{eqn:inhom_general_diff_eq}, most existing quantum algorithms convert the problem into a quantum linear system problem (QLSP) with a fixed and dilated matrix (i.e., a matrix of enlarged size). The resulting QLSP can then be solved using many of the aforementioned techniques based on the spectral mapping argument. However, both the construction of the linear system problem and the solution of the QLSP with near-optimal complexity (in order to achieve desired dependence on parameters such as the precision and the simulation time) can be very involved. We shall compare our new proposals with the QLSP approach later in the paper.

In this work, we propose a significantly simplified solution to the non-unitary process in \cref{eqn:inhom_general_diff_eq}. Our procedure is not based on the spectral mapping theorem, but  on a surprising identity expressing the solution as a linear combination of Hamiltonian simulation (LCHS) problems. LCHS can be viewed as a special case of LCU. Namely, each Hamiltonian simulation problem is described by a unitary operator. Unlike LCU for Hamiltonian simulation or solving linear systems, these unitary operators do not commute with each other. LCHS is a very flexible procedure. The linear combination can be implemented in a hybrid quantum-classical fashion to compute observables related to the solution, using a small amount of quantum resources and is thus suitable for the setting of early fault-tolerant quantum computers. The linear combination can also be coherently implemented to prepare the solution directly in a quantum register and to reduce the complexity. In this case, we show that the cost of LCHS is optimal in terms of state preparation, which is useful when the initial state $u_0$ is difficult to prepare. 

When the anti-Hermitian part of the matrix $A(t)$ is fast-forwardable, we may incorporate the interaction picture Hamiltonian simulation in the LCHS to obtain further improvements. 
A practical example with such a feature is the open quantum system dynamics. 
Many problems in quantum dynamics, such as molecular scattering~\cite{MahapatraSathyamurthy1997}, photodissociation~\cite{VibokBalint-Kurti1992}, and nanotransport~\cite{Datta2005}, are defined in an infinite space.
Unlike solving the ground state of molecules in quantum chemistry, replacing the infinite space by a finite-sized box in these quantum dynamics problems may lead to significant errors at least along certain extended dimensions. 
Therefore boundary conditions need to be carefully designed and implemented to balance the accuracy of the simulation and the computational cost. 
One widely used method in quantum chemistry is the complex absorbing potential (also called the imaginary potential) method~\cite{Child1991,VibokBalint-Kurti1992,MugaPalaoNavarroEtAl2004}. 
In this case, we show that our LCHS algorithm can achieve near-optimal dependence on \textit{all} parameters in both state preparation and matrix input models. 

\medskip
\prlsection{Linear combination of Hamiltonian simulation}
Let us consider the homogeneous problem first with $b(t)=0$.
\begin{thm}[Linear combination of Hamiltonian simulation]
Let $A(t)\in \CC^{N\times N},t\in\mc{I}=[0,T]$ be decomposed into a Hermitian and an anti-Hermitian part, 
$A(t)=L(t)+iH(t)$, where $L(t)=\frac{A(t)+A^{\dag}(t)}{2}$ and $H(t)=\frac{A(t)-A^{\dag}(t)}{2i}$. 
Assume $L(t)\succeq 0$ for all $t\in\mc{I}$.  Then ($\mc{T}$ is the time ordering operator)
\begin{equation}
\mc{T}e^{-\int_0^t A(s)\ud s} = \int_{\RR} \frac{1}{\pi(1+k^2)} \mc{T} e^{- i \int_0^t (H(s)+kL(s)) \ud s} \ud k. 
\label{eqn:lchs}
\end{equation}
\label{thm:lchs}
\end{thm}

\cref{thm:lchs} can be viewed as a generalization of the Fourier representation of the exponential function  $f(x)= e^{-\abs{x}}$
\begin{equation}
\hat{f}(k)=\frac{1}{2\pi}\int_{\RR} e^{-\abs{x}} e^{-i kx} \ud x= \frac{1}{\pi(1+k^2)}.
\label{eqn:lchs_scalar}
\end{equation}
Note that $\hat{f}(k)\ge 0$ and $\int \hat{f}(k)\ud k=1$. Therefore $\hat{f}(k)$ is the density of a probability distribution, called the Cauchy--Lorentz distribution. If $H(t)=0$ and $L(t)=L$ is time-independent, \cref{thm:lchs} can be readily proved from \cref{eqn:lchs_scalar} and the spectral mapping theorem. 
\REV{This special-case formula has been applied to simulating imaginary time evolution dynamics~\cite{ZengSunYuan2022,HuoLi2023}.
However, our~\cref{thm:lchs} works in a more general setting where the matrix can be time-dependent and non-Hermitian. }
\REV{Our general proof hinges on a special instance of the matrix version of the Cauchy integral theorem, which is a key for avoiding the spectral mapping argument (see the Supplemental Materials).}

The condition that the Hermitian part $L(t)$ is positive semidefinite can always be satisfied without loss of generality. Indeed, by redefining $u(t)=e^{ct} v(t)$, the equation for $v(t)$ is
\begin{equation}
\partial_t v(t) = -(L(t)+cI+iH(t)) v(t)+e^{-ct}b(t), \quad v(0)=u_0.
\end{equation}
By choosing $-c$ to be the minimum of the smallest eigenvalues of $L(t)$ on $t\in \mc{I}$, $L(t)+cI$ is a positive semidefinite matrix.

When $b(t)=0$, the solution to \cref{eqn:inhom_general_diff_eq} becomes $u(t)=\mc{T}e^{-\int_0^t A(s)\ud s} u_0$. By discretizing the integral with respect to $k$ using a grid $k_j$ with quadrature weights $\omega_j$, \cref{eqn:lchs} becomes
\begin{equation}
u(t) \approx  \sum_{j} c_j U_j(t)u_0, 
\end{equation}
where $c_j=\frac{\omega_j}{\pi(1+k_j^2)}$, and $U_j(t)=\mc{T} e^{- i \int_0^t (H(s)+k_j L(s)) \ud s}$ is the propagator for a time-dependent Hamiltonian simulation problem.
 Therefore \cref{eqn:lchs} compactly expresses the solution as a problem of LCHS, which can be coherently implemented using LCU \REV{(see the Supplemental Materials).}

If we are only interested in obtaining observables of the form $u(t)^{*} O u(t)$, we may implement the linear combination in a hybrid quantum-classical fashion. Notice that, since $U_j$'s are unitary, the observable can be expressed as 
\begin{equation}
u(t)^{*} O u(t) \approx \sum_{k,k'} c^*_k c_{k'}\braket{u_0|U^{\dagger}_k(t)OU_{k'}(t)|u_0}.
\label{eqn:observable}
\end{equation} 
We can then use the quantum computer to evaluate a series of correlation functions $\braket{u_0|U^{\dagger}_k(t)OU_{k'}(t)|u_0}$ via \REV{the Hadamard test for non-unitary matrices~\cite{TongAnWiebe2021}} and amplitude estimation~\cite{BrassardHoyerMoscaEtAl2002}, and perform the summation on a classical computer 
\REV{(see the Supplemental Materials). }

In the presence of the source term $b(t)$, we can use the Duhamel's principle (a.k.a. variation of constants) to express the solution as
\begin{equation}\label{eqn:Duhamel}
\begin{split}
    u(t) &=\int_{\RR} \frac{1}{\pi(1+k^2)} \mc{T} e^{- i \int_0^t (H(s)+kL(s)) \ud s} u_0\ud k \\
    & + \int_0^t \int_{\RR} \frac{1}{\pi(1+k^2)} \mc{T} e^{- i \int_s^t (H(s')+kL(s')) \ud s'} b(s) \ud k\ud s.
\end{split}
\end{equation}
We may again use LCU to coherently prepare the state $u(t)$, or use an expression similar to that of \cref{eqn:observable} for hybrid computation of observables. 

\medskip
\prlsection{Implementation}
The LCHS can be implemented in a gate-efficient way by combining LCU with any Hamiltonian simulation algorithms. 
Here we discuss the simplest implementation of LCHS based on the product formula, and we will discuss the one based on the truncated Dyson series method later in the paper. 

We first truncate the integral in~\cref{eqn:lchs} on a finite interval $[-K,K]$ and discretize it by a trapezoidal rule with $(M+1)$ grid points. 
We obtain $\mathcal{T} e^{-\int_0^T A(s) \ud s}u_0 \approx \sum_{j=0}^M c_j \mathcal{T} e^{-i\int_0^T (H(s) + k_jL(s)) \ud s} u_0$. 
Here $c_j = \frac{w_j}{\pi(1+k_j^2)}$, $w_j = \frac{(2-\mathbf{1}_{j=0,M})K}{M}$ and $k_j = -K+\frac{2jK}{M}$ are the weights and nodes of the trapezoidal rule.  
To implement each $U_j(T)=\mathcal{T} e^{-i\int_0^T (H(s)+k_jL(s)) \ud s}$, we use a $p$-th order product formula with a fixed number of steps $r$ for all $j$. 
Then, 
\begin{equation}\label{eqn:LCU_Trotter_homo}
\begin{split}
    &\mathcal{T} e^{-\int_0^T A(s) \ud s} u_0 \approx \sum_{j=0}^M c_j v_j, \\
    &v_j = \prod_{l'=0}^{r-1} \prod_{l=0}^{\Xi_p-1} \left(e^{-i H\left(\frac{(l'+ \delta_l) T}{r}\right)\frac{\beta_l T}{r}} e^{-iL \left(\frac{(l'+ \gamma_l) T}{r}\right) \frac{\alpha_l k_j T}{r}}\right) u_0. 
\end{split}
\end{equation}
Here $\Xi_p$ is the number of the exponentials in the product formula, $\alpha_l$'s and $\beta_l$'s are the corresponding coefficients, and $\gamma_l$'s and $\delta_l$'s determine the discrete times at which the time-dependent Hamiltonians are evaluated (see~\cite{WiebeBerryHoyerEtAl2010} for an example of the product formula via Suzuki recursion). 

Suppose that we are given the state preparation oracle $O_{\text{prep}}: \ket{0}\rightarrow \ket{u_0}$, the Hamiltonian simulation oracles $O_{L}(s,\tau) = e^{-i L(\tau ) s}$ for $|s| \leq 1/\|L\|$ and $O_{H}(s,\tau) = e^{-i H(\tau ) s}$ for $|s| \leq 1/\|H\|$, and the LCU coefficient oracle  $O_{\text{coef}}: \ket{0} \rightarrow \frac{1}{\sqrt{\|c\|_1}} \sum_{j=0}^{M} \sqrt{c_j} \ket{j}$. 
According to the binary representation of $j$'s, we may first construct a coherent encoding of the time evolution as the select oracle $\text{SEL}_L (s,\tau) = \sum_{j=0}^M \ket{j}\bra{j} \otimes e^{-iL (\tau ) k_j s}$ using $\mathcal{O}(\log(M))$ queries to $O_{L}(s,\tau)$ (see Supplemental Materials and, e.g.,~\cite{ChildsKothariSomma2017,LowWiebe2019,AnFangLin2022} for details). 
Then~\cref{eqn:LCU_Trotter_homo} can be implemented following the standard LCU approach by first applying $O_{\text{coef}} \otimes O_{\text{prep}}$, then sequentially applying $\text{SEL}_L(\alpha_l T/r,(l'+\gamma_l)T/r)$ and  $O_H(\beta_l T/r,(l'+\delta_l)T/r)$ for $l\in[\Xi_p]$ and $l'\in[r]$, and finally applying $O_{\text{coef}}^{\dagger}$ on the ancilla register. 
Such a procedure yields the quantum state approximating $\frac{1}{\|c\|_1\|u_0\|}\ket{0}_a \mathcal{T} e^{-\int_0^T A(s) \ud s} u_0 + \ket{\perp}$, encoding the homogeneous part in its first subspace.

For the inhomogeneous term in~\cref{eqn:Duhamel}, after discretizing the integral for both $k$ and $s$ using the multidimensional trapezoidal rule, we obtain $\int_0^T \mathcal{T} e^{-\int_s^T A(s') \ud s'} b(s) \ud s \approx \sum_{j'=0}^{M_t} \sum_{j=0}^M  \widetilde{c}_{j,j'} \mathcal{T} e^{- i \int_{s_{j'}}^T (H(s')+k_j L(s')) \ud s' } \ket{b(s_{j'})}$. 
Here $\widetilde{c}_{j,j'} = \frac{v_{j'}w_j \|b(s_{j'})\|}{\pi(1+k_j^2)}$, $w_j$, $k_j$ are as defined before, $v_{j'} = \frac{(2-\mathbf{1}_{j'=0,M})T}{2M_t}$ and $s_{j'} = \frac{j'T}{M_t}$. 
This can also be implemented by the same Trotterization and LCU approach as the homogeneous case (see Supplemental Materials for details). 
Notice that the evolution time of different Hamiltonians in the LCHS of the inhomogeneous term varies, so we will assume the input model of $H$, $L$ to coherently encode the evolution of different Hamiltonians with different time periods, as $O_L'(s,\tau_0,\tau_1) = \sum_{j'=0}^{M_t} \ket{j'}\bra{j'}\otimes  e^{ -iL(\tau_0 j' + \tau_1 (M_t-j'))s(M_t-j')}$ and 
$O_H'(s,\tau_0,\tau_1) = \sum_{j'=0}^{M_t} \ket{j'}\bra{j'} \otimes e^{-iH(\tau_0 j' + \tau_1 (M_t-j'))s(M_t-j')}$ for $|s|M_t \leq 1/\|H\|$. 
The input for $b(t)$ is also in a time-dependent manner as $O_{b}: \ket{j'}\ket{0} \rightarrow \ket{j'}\ket{b(s_{j'})}$. 
All of these oracles are an extension of the time-dependent encoding proposed in~\cite{LowWiebe2019}. 

To linearly combine $\mathcal{T}e^{- \int_0^T A(s) \ud s}u_0$ and $\int_0^T \mathcal{T}e^{- \int_s^T A(s') \ud s'} b(s) \ud s$, we append an extra ancilla qubit, prepare both states controlled by this ancilla qubit, and implement the LCU again at the outer loop using a single-qubit rotation. 
The final step is to measure all the ancilla registers, and if all the outcomes are $0$, then the resulting state approximately encodes the solution $u(t)$ of the ODE. 

\medskip
\prlsection{Computational cost}
The complexity of our algorithm for solving~\cref{eqn:inhom_general_diff_eq} is given as follows.  Here $\ket{u(T)}$ is the solution to \cref{eqn:inhom_general_diff_eq} at the final time $T$.

\begin{thm}\label{thm:td_inhomo}
    There exists a quantum algorithm that prepares an $\epsilon$-approximation of the state $\ket{u(T)}$  with $\Omega(1)$ success probability and a flag indicating success, using 
    \begin{enumerate}
        \item queries to the aforementioned input models of $H$ and $L$ a total number of times 
        \begin{equation}
            \widetilde{\mathcal{O}}\left( \left( \frac{\|u_0\|+\|b\|_{L^1}}{\norm{u(T)} } \right)^{2+2/p} \frac{ \Gamma_p^{1+1/p} T^{1+1/p}}{\epsilon^{1+2/p}} \right)
        \end{equation}
        where $\|b\|_{L_1} = \int_0^T \|b(s)\| \ud s$ and $\Gamma_p = \max_{0\leq q \leq p, \tau \in [0,T]} \left( \|H^{(q)}(\tau)\|+ \|L^{(q)}(\tau)\|\right)^{1/(q+1)}$, 
        \item queries to the state preparation oracle $O_{\text{prep}}$ and the source term input model $O_b$ for $\mathcal{O}\left( \frac{ \|u_0\|+ \|b\|_{L^1}}{\norm{u(T)} } \right)$ times,  
        \item $\mathcal{O}(\log(\Gamma_1 \|b\|_{C^2} T/\epsilon))$ ancilla qubits, where $\|b\|_{C^2} = \sum_{q=0}^2 \max_{\tau \in [0,T]} \|b^{(q)}(\tau)\| $, 
        \item $\mathcal{O}\left( \frac{ \|u_0\|+ \|b\|_{L^1}}{\norm{u(T)} } \right)$ additional one-qubit gates. 
    \end{enumerate}
\end{thm}

The proof of~\cref{thm:td_inhomo} can be found in the Supplemental Materials. 
Our algorithm uses a small number of queries to the state preparation oracle. 
This is because each run of the LCU procedure only requires $\mathcal{O}(1)$ uses of such oracles and the overall complexity only relates to the success probability, which contributes to a $(\|u_0\|+\|b\|_{L^1})/\|u(T)\|$ factor. 
The query complexity to the matrix input in each run depends linearly on the number of Trotter steps, which contributes to the scaling $\Gamma_p^{1+1/p} T^{1+1/p}/\epsilon^{1/p}$ according to the $p$-th order Trotter error bound. 
The extra $1/\epsilon^{1/p}$ scaling is due to the relative error scaling. 
Notice that in the matrix query complexity there are still extra terms including $1/\epsilon$ and $((\|u_0\|+\|b\|_{L^1})/\|u(T)\|)^{1+1/p}$. 
The former is because we need to simulate the Hamiltonian up to $K = \Or(\epsilon^{-1})$, while the latter arises from the necessity of bounding the relative Trotter error.
Ancilla qubits are used to implement quadrature.

Our algorithm and~\cref{thm:td_inhomo} can be directly applied to the special case where $A(t) \equiv A$ is time-independent. In this case, we may further simplify the implementation and reduce the computational cost. 
First, the coherent encoding of the time evolution is no longer needed, and the select oracles in the LCU procedure can be efficiently constructed using $O_L(s) = e^{-iLs}$ and $O_H(s) = e^{-iHs}$ with $\mathcal{O}(\log(M)\log(M_t))$ cost. 
Second, the parameter $\Gamma_p$ can be improved to the commutator scalings between $H$ and $L$ thanks to the improved error bound for time-independent product formula~\cite{ChildsSuTranEtAl2020}. 
We refer to the Supplemental Materials for more details. 

\medskip
\prlsection{Optimal state preparation cost and comparison with other methods}
The query complexity of the LCHS approach to the state preparation oracle $O_{\text{prep}}$ is $\mathcal{O}\left( \frac{\|u_0\|+\|b\|_{L^1}}{\norm{u(T)} } \right)$. 
In fact, when $b=0$, this corresponds to the dependence of a quantity denoted by $q = \frac{\|u_0\|}{\norm{u(T)} }$ in quantum differential equation solvers. Such a linear scaling of $q$ cannot be improved. From the perspective of quantum complexity theory, the renormalization from $\|u_0\|$ to $\norm{u(T)}$ could be utilized to implement postselection, and it would imply the unlikely consequence \texttt{BQP = PP} (similar discussions appear in Section 8 of~\cite{BerryChildsOstranderEtAl2017} and Section 7 of~\cite{LiuKoldenKroviEtAl2021}). 
Moreover, as shown in ~\cite[Corollary 16]{an2022theory}, any quantum differential equation algorithm must have worst-case query complexity $\Omega(q)$ to the state preparation oracle, indicating the optimal state preparation cost that our LCHS approach achieves.

Most generic quantum differential equation algorithms~\cite{Berry2014,BerryChildsOstranderEtAl2017,ChildsLiu2020,Krovi2022,berry2022quantum} convert the time-dependent differential equation problem \cref{eqn:inhom_general_diff_eq} into a QLSP. 
The efficiency of these quantum algorithms relies on the efficiency of the quantum linear system algorithms (QLSA), which typically take a large number of queries to the state preparation oracle. 
In particular, the state-of-the-art QLSA takes $\Or(\kappa\log(1/\epsilon))$ queries~\cite{CostaAnYuvalEtAl2022}, where $\kappa$ is the condition number of the QLSP matrix and depends on $T$ and $\|A(t)\|$. 
However, our algorithm directly implements the time evolution operator by LCHS without using QLSA, so the number of the state preparation oracle is significantly reduced. 
For example, when $b=0$, our algorithm takes $\mathcal{O}\left(\frac{\norm{u_0}}{\norm{u(T)}}\right)$ queries to the state preparation, while the state-of-the-art QLSA-based quantum Dyson series method~\cite{berry2022quantum} queries the initial state $\mathcal{O}\left( \frac{\norm{u_0}}{\norm{u(T)} }\|A\| T \log(1/\epsilon)\right)$ times. 
Our algorithm removes the explicit dependence on $\|A\|$, $T$ and $\epsilon$, and matches the lower bound. 
Additionally, in the QLSA-based approaches, the value of $\kappa$ can be difficult to estimate in practice, and it is common that the theoretical bound significantly overestimates the value of $\kappa$~\cite{Krovi2022}. 
We also note that the time marching method~\cite{FangLinTong2022} does not involve QLSP and can query the initial state $\mathcal{O}\left( \frac{\|u_0\|}{\norm{u(T)} } \right)$ times as well. 
In this case, our LCHS approach removes the need of implementing the uniform singular value amplification procedure and can thus be simpler to implement. It also outperforms the time marching method in terms of the number of matrix queries with respect to the evolution time, with an improvement from $\Or(T^2)$ to $\Or(T^{1+o(1)})$. 

\medskip
\prlsection{Applications to open quantum system dynamics with complex absorbing potentials}
 In its simplest form, the complex absorbing potential method~\cite{Child1991,VibokBalint-Kurti1992,MugaPalaoNavarroEtAl2004} replaces the real potential by a complex one, and the time-dependent Schr\"odinger equation becomes
\begin{equation}
i\partial_t u(\vr,t)=\left(-\frac12 \Delta_{\vr} +V_R(\vr,t)-iV_I(\vr)\right)u(\vr,t). 
\label{eqn:complex_absorb}
\end{equation}
Here $V_R(\vr,t)$ is the real time-dependent external potential, and $-iV_I(\vr)$ is the absorbing potential and can often be chosen to be time-independent. The minus sign reflects that this is a damping potential, and $V_I$ can be chosen to be bounded and non-negative. 
For simplicity we only discuss \cref{eqn:complex_absorb} in the context of single-particle dynamics, and this formulation can be generalized to accommodate multi-particle dynamics. 
Our method can also be generalized to other more boundary treatment methods such as the perfectly matched layer (PML) method~\cite{Berenger1994,Zheng2007}. In these applications, we are interested in the case before the scattering wave leaves the region of interest, i.e.,  $\norm{u(T)}^2=\int \abs{u(\vr,T)}^2 \ud \vr$ is not too small. Such a condition can also be satisfied if $u_0$ is a near-resonance state. \REV{Another widely used method for modeling open quantum system dynamics is the Gorini--Kossakowski--Sudarshan--Lindblad (GKSL) quantum master equation~\cite{Lindblad1976,GoriniKossakowskiSudarshan1976,LandiPolettiSchaller2022}. The non-Hermitian quantum dynamics \eqref{eqn:complex_absorb} is also related to a class of numerical methods for solving the GKSL equation, known as the quantum jump or the Monte Carlo wavefunction method~\cite{DalibardCastinMolmer1992,DumZollerRitsch1992}.}

To solve~\cref{eqn:complex_absorb} using our LCHS algorithm, we first discretize the spatial variable using $N$ equidistant grid points, and the Laplace operator is discretized by the central difference formula. 
In the context of \cref{eqn:inhom_general_diff_eq}, the Hermitian part $H(t)=-\frac12 \Delta_{\vr} +V_R(t)$ is the standard Hamiltonian with $\|H(t)\| = \mathcal{O}(N^2+\max_t\|V_R(t)\|)$, and the non-Hermitian part $L=V_I$ is a time-independent positive semi-definite matrix. 
Furthermore, the dynamics $e^{-ikLt}$ can be fast-forwarded in the sense that $e^{-ikLt}$ can be performed with cost independent of $k,t$, and $\|L\|$~\cite{ahokas2004improved}, so we may perform the Hamiltonian simulation in the interaction picture with the truncated Dyson series method~\cite{LowWiebe2019} to avoid computational overhead brought by $K$. 

In the interaction picture, the LCHS becomes 
\begin{equation}\label{eqn:LCHS_interaction_picture}
    \begin{split}
        \mathcal{T} e^{-\int_0^t A(s) \ud s } \approx \sum_{j=0}^{M} c_j e^{-i L k_j t} \left(\mc{T} e^{- i \int_0^t H_I(s;k_j) \ud s}\right)  e^{i L k_j t}, 
    \end{split}
\end{equation}
where $H_I(s;k) = e^{iL k s} H(s) e^{-i L k s}$. 
\cref{eqn:LCHS_interaction_picture} can be directly implemented by LCU. 
Notice that the derivative of $H_I$ still scales linearly in $K$, but the cost of the truncated Dyson series method scales only logarithmically with respect to the derivative of $H_I(s)$ and thus $K$~\cite{LowWiebe2019}. 

Our main result is as follows, which achieves near-optimal scaling in \textit{all} parameters. 
The proof is given in the Supplemental Materials. 
We remark that the following result also holds in a more general case where $H(t)$ is an arbitrary time-dependent Hamiltonian and $L$ is a time-independent fast-forwardable Hamiltonian. 

\begin{thm}\label{thm:cap}
    Consider the spatially discretized~\cref{eqn:complex_absorb} using finite difference with $N$ equidistant grid points. 
    Suppose that we are given the oracles $O_{V_I}: \ket{\vr}\ket{0} \rightarrow \ket{\vr}\ket{V_I(\vr)}$,  $O_{V_R}: \ket{\vr}\ket{s}\ket{0} \rightarrow \ket{\vr}\ket{s}\ket{V_R(\vr,s)}$, and the state preparation oracle $O_{\text{prep}}$ for the initial condition. 
    Then there exists a quantum algorithm that prepares an $\epsilon$-approximation of the state $\ket{u(T)}$ with $\Omega(1)$ success probability and a flag indicating success, using 
    \begin{equation}
        \widetilde{\mathcal{O}} \left( \frac{\norm{u_0}}{\norm{u(T)}}  T (\max_t\|H(t)\|)\text{~poly}\log\left(\frac{\max_t\|V_R'(t)\|\|V_I\|}{\epsilon} \right) \right)
    \end{equation}
    queries to $O_{V_I}$ and $O_{V_R}$, and $\mathcal{O}(\norm{u_0}/\norm{u(T)})$ queries to $O_{\text{prep}}$. 
    Here $\max_t\|H(t)\| = \mathcal{O}(N^2+\max_t\|V_R(t)\|)$. 
\end{thm}

\medskip
\prlsection{Discussion}
Linear combination of Hamiltonian simulation (LCHS) provides a simple way for solving  the non-unitary dynamics in the form of \cref{eqn:inhom_general_diff_eq}. Compared to existing approaches based on QLSP, LCHS is simple, flexible, and achieves optimal state preparation cost. The linear combination procedure can be implemented in a hybrid quantum-classical fashion to facilitate the computation of observables on early fault-tolerant quantum computers. 
The most significant limitation of LCHS is that the spectral radius of the Hermitian part $L$ needs to be multiplied by a factor up to the frequency cutoff $K$, where $K=\Or(1/\epsilon)$ due to the quadratic decay of the kernel $(1+k^2)^{-1}$. This increases the maximal circuit depth by a factor $\epsilon^{-1}$. The problem can be  ameliorated when the dynamics of $L$ can be fast-forwarded. For instance, in the case of simulating open quantum system dynamics using a complex absorbing potential, our LCHS-based solver achieves near optimal complexity in all parameters. However, Hamiltonian simulation in the interaction picture~\cite{LowWiebe2019} may be difficult to implement.  A more desirable solution to this problem would be to replace the kernel $(1+k^2)^{-1}$ by a fast decaying one, so that the truncation range can be dramatically reduced. It remains an open question how to obtain the best possible solver that is optimal in all parameters for simulating general non-unitary dynamics.

When a task can be solved using either LCU- or QSP-based techniques, the latter are often preferred due to their superior resource efficiency, as well as ease of implementation. However, all QSP-based techniques are based on the spectral mapping theorem. LCHS avoids this spectral mapping argument, and thus significantly expands the application range of LCU. It would be interesting to see if the  mathematical reformulation of LCHS can inspire further generalization of QSP-based approaches, particularly for applications involving non-normal matrices. Additionally, exploring  extensions of LCHS
to other quantum linear algebra problems can be a promising direction for future research.

\begin{acknowledgments}
We thank Dominic Berry, Andrew Childs, Zhiyan Ding, Di Fang and Yu Tong for helpful discussions and suggestions. This material is based upon work supported by the U.S. Department of Energy, Office of Science, National Quantum Information Science Research Centers, Quantum Systems Accelerator (L.L.). D.A. acknowledges the support by the Department of Defense through the Hartree Postdoctoral Fellowship at QuICS. J.-P.~L.~acknowledges the support by the NSF (CCF-1813814, PHY-1818914), the NSF QLCI program (OMA-2016245), and the Simons Quantum Postdoctoral Fellowship. L.L. has received support as a Simons investigator. 

\textbf{Note:} During the final stage of this work, we became aware of a recent work called \textit{Schr\"odingerisation}~\cite{JinLiuYu2022}, which transforms a general class of linear partial differential equations into a dilated Hamiltonian dynamics, and can be used to provide a complementary perspective of our result in \cref{thm:lchs}. 
\end{acknowledgments}

\bibliography{LCH}

\include{supp_prl}

\end{document}

%% file: supp_prl.tex
\newpage 
\clearpage
\thispagestyle{empty}
\onecolumngrid
\begin{center}
\textbf{\large Supplemental Material for \\ Linear combination of Hamiltonian simulation for non-unitary dynamics with optimal state preparation cost }
\end{center}

\begin{center}
Dong An,$^1$ Jin-Peng Liu,$^{2,3,4}$ and Lin Lin$^{2,5,6}$\\
\smallskip
\small{\emph{$^1$\QUICS\\$^2$\DeptMath\\$^3$\Simons\\$^4$\CTP\\$^5$\LBLMath\\$^6$\CIQC}}\\
(Dated: \today)

\end{center}

\setcounter{equation}{0}
\setcounter{figure}{0}
\setcounter{table}{0}
\setcounter{page}{1}
\makeatletter
\renewcommand{\theequation}{S\arabic{equation}}
\renewcommand{\thefigure}{S\arabic{figure}}
\renewcommand{\bibnumfmt}[1]{[S#1]}
\renewcommand{\citenumfont}[1]{#1}

\section{Proof of \MakeLowercase{\texorpdfstring{\cref{thm:lchs}}{} }}\label{eqn:proof_lchs}

We first establish a lemma which is a special instance of Cauchy's theorem applied to matrix functions without relying on the spectral mapping argument.

\begin{figure}[H]
\begin{center}
\includegraphics[width=0.3\textwidth]{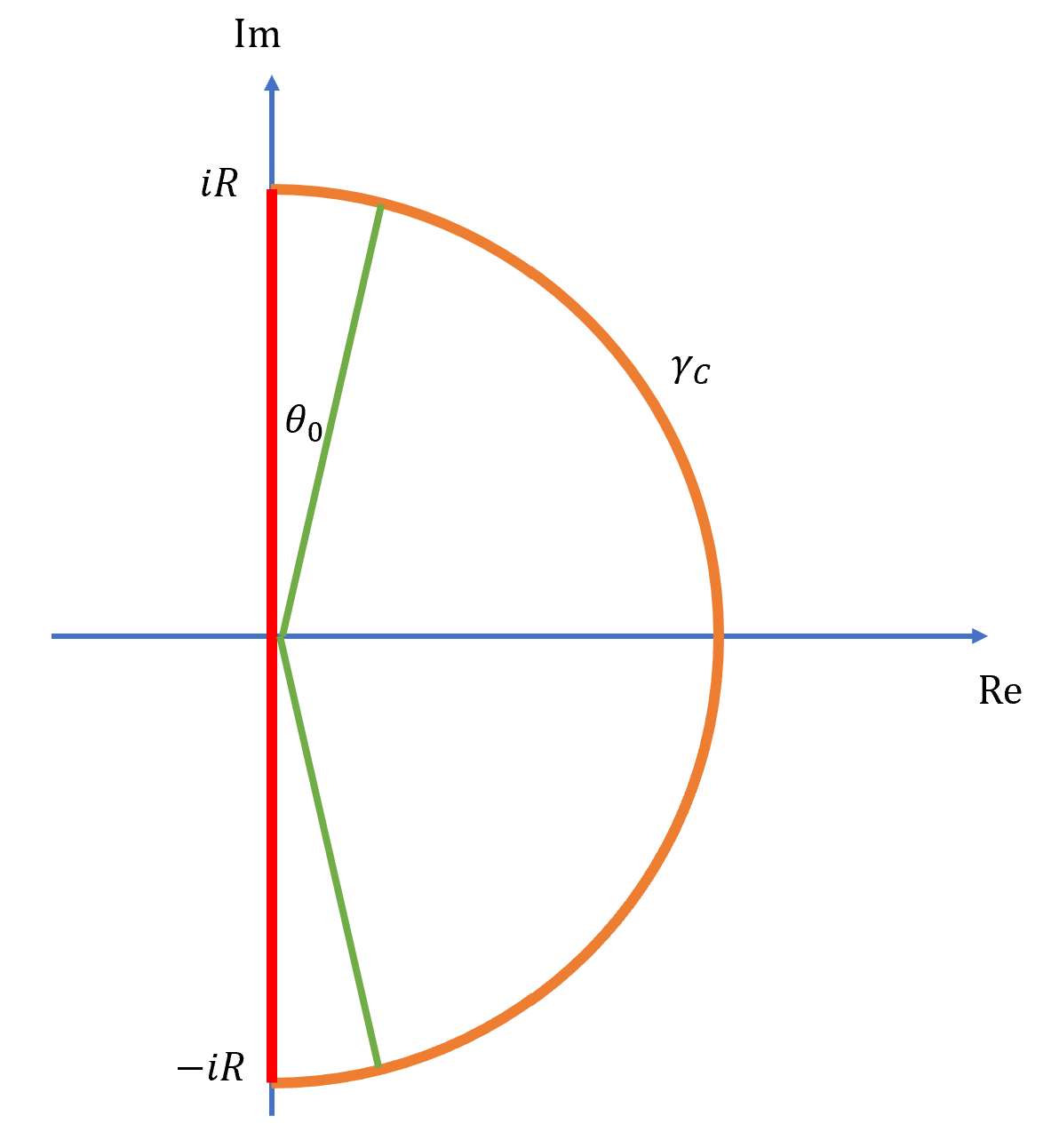}
\end{center}
\caption{Contour used to prove \cref{lem:cauchy_matrix}.}
\label{fig:contour_jordan}
\end{figure}

\begin{lem} Let $H,L\in\CC^{N\times N}$ be Hermitian matrices, and $L\succ 0$. Then
\begin{equation}
\mathcal{P}\int_{\RR}\frac{1}{1+ik} e^{-i(H+kL)} \ud k := \lim_{R\to \infty} \int_{-R}^R\frac{1}{1+ik} e^{-i(H+kL)} \ud k = 0.
\label{eqn:integral_cancel}
\end{equation}
Here $\mathcal{P}$ stands for the Cauchy principal value of the integral.
\label{lem:cauchy_matrix}
\end{lem}
\begin{proof}

By a change of variable $\omega=ik$, and denote the imaginary axis by $\gamma$, we have 
\begin{equation}
\mathcal{P}\int_{\RR}\frac{1}{1+ik} e^{-i(H+kL)} \ud k=\mathcal{P}\int_{\gamma}\frac{1}{i(1+\omega)} e^{-iH-\omega L} \ud \omega.
\label{eqn:cancel_temp1}
\end{equation}
Let $R$ be any positive number, 
and we choose a closed contour $\mc{C}$ in the complex plane as in \cref{fig:contour_jordan},
where $\gamma_C$ is the half circle with a radius $R$ in the right half plane. We will prove the lemma by deforming the integral along the imaginary axis to the integral on along the half circle.

Note that \textit{each entry} of the integrand on the right hand side of \cref{eqn:cancel_temp1} is analytic with respect to $\omega$ in the region enclosed by $\mathcal{C}$. Therefore Cauchy's integral theorem applies, and 
\begin{equation}\label{eqn:cancel_temp5}
\mathcal{P}\int_{\gamma}\frac{1}{i(1+\omega)} e^{-iH-\omega L} \ud \omega=\lim_{R\to \infty} \int_{-iR}^{iR} \frac{1}{i(1+\omega)} e^{-iH-\omega L} \ud \omega=
\lim_{R\to \infty} \int_{\gamma_C} \frac{1}{i(1+\omega)} e^{-iH-\omega L} \ud \omega.
\end{equation}
We can parameterize $\gamma_C=\set{R e^{i\theta}|\theta\in[-\pi/2,\pi/2]}$, which can be separated into two parts. Choose $\theta_0=\min(1/\sqrt{R},\pi/4)$. The first part is $\theta\in I=[-\pi/2+\theta_0,\pi/2-\theta_0]$, and the second is $\theta\in J=[-\pi/2,\pi/2]\backslash I$.  We also denote by $\lambda_0>0$ the smallest eigenvalue of $L$.

Changing the variable from $\omega$ to $\theta$ in~\cref{eqn:cancel_temp5} gives 
\begin{equation}
    \mathcal{P}\int_{\gamma}\frac{1}{i(1+\omega)} e^{-iH-\omega L} \ud \omega 
    = \lim_{R\to \infty} \int_{-\pi/2}^{\pi/2} \frac{R e^{i\theta}}{1+R e^{i\theta}} e^{-iH-R e^{i\theta} L} \ud \theta = \lim_{R\to \infty} \int_{I\cup J} \frac{R e^{i\theta}}{1+R e^{i\theta}} e^{-iH-R e^{i\theta} L} \ud \theta. 
\end{equation}
On the interval $I$, we have 
\begin{equation}
\norm{\int_{I} \frac{R e^{i\theta}}{1+R e^{i\theta}} e^{-iH-R e^{i\theta} L} \ud \theta}\le \int_{I} \frac{1}{\abs{1+R^{-1}e^{-i\theta}}} e^{- \lambda_0 R \cos\theta } \ud \theta.
\end{equation}
Notice that $R\cos\theta\ge R\sin\theta_0\ge \frac{2\sqrt{R}}{\pi}$. 
Therefore this term vanishes in the limit $R\to \infty$. 

On the interval $J$, for $R \geq 2$, we have
\begin{equation}
\norm{\int_{J} \frac{R e^{i\theta}}{1+R e^{i\theta}} e^{-iH-R e^{i\theta} L} \ud \theta}\le \int_{J} \frac{1}{\abs{1+R^{-1}e^{-i\theta}}} \ud \theta\le 2\theta_0 \le \frac{2}{\sqrt{R}},
\end{equation}
which vanishes as $R\to \infty$. This proves the lemma. 
\end{proof}

Denote the right hand side of \cref{eqn:lchs} by $V(t;ik)=\mc{T} e^{- i \int_0^t (H(s)+kL(s)) \ud s}$. Then $V$ is analytic with respect to $\omega=ik$ in the complex plane for any $t$. Furthermore, if the eigenvalues of $L(t)$ are uniformly bounded from below by $\lambda_0>0$, then for $\theta\in I=[-\pi/2+\theta_0,\pi/2-\theta_0]$, 
\begin{equation}
\norm{V(t,Re^{i\theta})}\le e^{- t\lambda_0 R \cos\theta}  
\end{equation}
still holds. This proves the following statement.

\begin{lem} Let $H(t),L(t)\in\CC^{N\times N}$ be Hermitian matrices on $t\in\mc{I}=(0,T]$. Assume $L(t)\succeq \lambda_0>0$ for all $t\in \mc{I}$. Then for any $0<t\le T$,
\begin{equation}
\mathcal{P}\int_{\RR}\frac{1}{1+ik} \mc{T} e^{- i \int_0^t (H(s)+kL(s)) \ud s} \ud k  = 0.
\label{eqn:integral_cancel_td}
\end{equation}
Here $\mathcal{P}$ stands for the Cauchy principal value of the integral.
\label{lem:cauchy_matrix_td}
\end{lem}

Now we are ready to establish the proof of~\cref{thm:lchs}. 

\begin{proof}[Proof of \cref{thm:lchs}]
We first prove the statement when $L(t)\succeq \lambda_0>0$ for all $t\in \mc{I}$.
The right hand of \cref{eqn:lchs} converges absolutely. Hence
\begin{equation}
 W(t):=\int_{\RR} \frac{1}{\pi(1+k^2)} \mc{T} e^{- i \int_0^t (H(s)+kL(s)) \ud s} \ud k=\mc{P}\int_{\RR} \frac{1}{\pi(1+k^2)} \mc{T} e^{- i \int_0^t (H(s)+kL(s)) \ud s} \ud k.
\end{equation}

When $t=0$, both sides of \cref{eqn:lchs} are the identity matrix due to the normalization condition of the Cauchy distribution. For $t>0$, the definition of the time-ordered operator gives\begin{equation}
\frac{\ud}{\ud t}V(t;ik)=-i (H(t)+kL(t))V(t;ik).
\end{equation}
We will show that $W(t)$ satisfies the same differential equation for $t > 0$. 
To this end, we choose a fixed $\delta > 0$ and consider $t \in [\delta,T]$. 
Differentiating $W(t)$ with respect to $t$ gives  
\begin{equation}\label{eqn:lchs_proof_intermediate}
\begin{split}
\frac{\ud W(t)}{\ud t}=&\mc{P}\int_{\RR} \frac{-i}{\pi(1+k^2)} (H(t)+kL(t))V(t;ik) \ud k\\
=&\mc{P}\int_{\RR} \left(\frac{(1-ik)L(t)}{2\pi(1+k^2)}-\frac{(1+ik)L(t)}{2\pi(1+k^2)} -\frac{iH(t)}{\pi(1+k^2)}\right) V(t;ik) \ud k\\
=&\mc{P}\int_{\RR} \left(\frac{L(t)}{2\pi(1+ik)}-\frac{L(t)}{2\pi(1-ik)} -\frac{iH(t)}{\pi(1+k^2)} \right)V(t;ik) \ud k.
\end{split}
\end{equation}
Here in~\cref{eqn:lchs_proof_intermediate}, since the right hand side of the first line uniformly converges in $t\in[\delta,T]$, we can exchange the order of limitation (\emph{i.e.}, the principal value) and differentiation, so the calculations in~\cref{eqn:lchs_proof_intermediate} rigorously hold. 
Now use \cref{lem:cauchy_matrix_td}, the first term vanishes and we can flip its sign as
\begin{equation}
\mc{P}\int_{\RR} \frac{L(t)}{2\pi(1+ik)}V(t;ik) \ud k=-\mc{P}\int_{\RR} \frac{L(t)}{2\pi(1+ik)}V(t;ik) \ud k=0.
\end{equation} 
Then, for $t \in [\delta,T]$, we have 
\begin{equation}\label{eqn:lchs_proof_intermediate_de}
\begin{split}
\frac{\ud W(t)}{\ud t}=&\mc{P}\int_{\RR} \left(\frac{-L(t)}{2\pi(1+ik)}-\frac{L(t)}{2\pi(1-ik)} -\frac{iH(t)}{\pi(1+k^2)} \right)V(t) \ud k\\
=&-(L(t)+iH(t))\int_{\RR} \frac{1}{\pi(1+k^2)} V(t;ik) \ud k.
\end{split}
\end{equation}
Since $\delta$ can be chosen arbitrarily close to $0$,~\cref{eqn:lchs_proof_intermediate_de} also holds for all $t \in (0,T]$. 
This proves that $W(t)$ satisfies the same differential equation as the left hand side of \cref{eqn:lchs}
\begin{equation}
\frac{\ud W(t)}{\ud t}=-(L(t)+iH(t))W(t)=-A(t)W(t), \quad W(0)=I.
\end{equation} 
By the uniqueness of the solution, we prove that $W(t)=\mc{T}e^{-\int_0^t A(s)\ud s}$ under for positive definite $L(t)$. 

Finally, each matrix entry of both the left and right hand sides of \cref{eqn:lchs} are continuous functions with respect to $L(t)$. This allows us to take the limit $\lambda_0\to 0$ and finishes the proof.
\end{proof}

\section{Numerical integration}
\label{app:quadrature}
We introduce the notation $\|H\| \coloneqq \max_{\tau \in [0,T]} \|H(\tau)\|_2$, $\|L\| \coloneqq \max_{\tau \in [0,T]} \|L(\tau)\|_2$, and $\|b\| \coloneqq \max_{\tau \in [0,T]} \|b(\tau)\|_2$. 
We further denote 
\begin{equation}
    \|b\|_{L^1} \coloneqq \int_{0}^T \|b(\tau)\|_{2} \ud \tau, \qquad \|b\|_{C^p} \coloneqq \sum_{q=0}^p \max_{\tau \in [0,T]}  \|b^{(q)}(\tau)\|_{2} ,
\end{equation}
\begin{equation}
    \|H\|_{C^p} \coloneqq \sum_{q=0}^p \max_{\tau \in [0,T]}  \|H^{(q)}(\tau)\|_{2} , \qquad \|L\|_{C^p} \coloneqq \sum_{q=0}^p \max_{\tau \in [0,T]}  \|L^{(q)}(\tau)\|_{2} . 
\end{equation}
We assume $H(t)$ and $L(t)$ are $C^1$-smooth, and $b(t)$ is $C^2$-smooth in the following analysis.

Firstly, we consider numerical integration formula to estimate
\begin{equation}
\int_{-K}^K F(k) \ud k = \int_{-K}^K \frac{1}{\pi(1+k^2)} \mathcal{T} e^{- i \int_0^t (H(s)+kL(s)) \ud s} \ud k.
\label{eqn:general_A_truncate}
\end{equation}
in which $F $ is a \textit{matrix-valued} function (but not a matrix function). We can choose the truncation parameter $K=c/\epsilon$ for some constant $c$, and the remainder is then bounded by 
\begin{equation}
2\int_K^\infty \frac{1}{1+k^2}\ud k=\pi-2\arctan K=\Or(K^{-1})=\Or(\epsilon).
\end{equation}

We can simply apply a trapezoidal rule on \cref{eqn:general_A_truncate}
\begin{equation}
\sum_{j=0}^M w_jF(k_j) = \sum_{j=0}^M c_j \mathcal{T} e^{-i\int_0^t (H(s) + k_jL(s)) \ud s}.
\end{equation}
Here $c_j = \frac{w_j}{\pi(1+k_j^2)}$, $w_j = \frac{(2-\mathbf{1}_{j=0,M})K}{M}$ and $k_j = -K+\frac{2jK}{M}$ are the weights and nodes of the trapezoidal rule. 

The global error of the trapezoidal rule is proportional to
\begin{equation}
\max_{\xi \in [-K,K]}\frac{(2K)^3}{M^2}|F''(\xi)| ,
\end{equation}
where 
\begin{equation}
|F''(k)| \le \frac{\|L\|^2T^2}{\pi(1+k^2)} + \frac{4k\|L\|T}{\pi(1+k^2)^2} + \frac{2(3k^2-1)}{\pi(1+k^2)^3}.
\end{equation}
Thus, $\max_{\xi} |F''(\xi)| = \Or(\|L\|^2T^2)$. In order to reach precision $\epsilon$, the total number of points is
\begin{equation}
M = \Or(\|L\|T/\epsilon^2).
\label{eqn:quadrature_homogeneous}
\end{equation}

Secondly, we consider numerical integration formula to estimate
\begin{equation}
\int_{-K}^K\int_0^T \widetilde F(k,s) \ud s \ud k = \int_{-K}^K\int_0^T \mathcal{T} \frac{1}{\pi(1+k^2)} e^{- i \int_s^t (H(s')+kL(s')) \ud s' } b(s) \ud s \ud k
\label{eqn:general_A_inhom_truncate}
\end{equation}
in the inhomogeneous case, spatially truncated on $[-K,K]$ with $K=\|b\|_{L^1}/\epsilon$, 
and the remainder is then bounded by 
\begin{equation}
2\int_0^T \int_K^\infty \frac{1}{1+k^2} \|b(s)\|_2\ud k \ud s = 2\int_0^T(\pi-2\arctan K)  \|b(s)\|_2 \ud s =\Or(\|b\|_{L^1}/K)=\Or(\epsilon).
\end{equation}

We can simply apply a trapezoidal rule on \cref{eqn:general_A_inhom_truncate}.
\begin{equation}
\sum_{j'=0}^{M_t} \sum_{j=0}^M w_{j,j'} \widetilde F(k_j, s_{j'}) = \sum_{j'=0}^{M_t} \sum_{j=0}^M \widetilde{c}_{j,j'} \mathcal{T} e^{- i \int_{s_{j'}}^t (H(s')+k_j L(s')) \ud s' } b(s_{j'})
\end{equation}
Here $\widetilde{c}_{j,j'} = \frac{w_{j,j'}}{\pi(1+k_j^2)}$, $w_{j,j'} = \frac{(2-\mathbf{1}_{j=0,M})(2-\mathbf{1}_{j'=0,M_t})KT}{2MM_t}$ and $k_j = -K+\frac{2jK}{M}$, $s_{j'} = \frac{j'T}{M}$ are the weights and nodes of the two-dimensional trapezoidal rule. 

As above, the global error of the trapezoidal rule is proportional to
\begin{equation}
2KT \max_{\xi \in [-K,K], \zeta \in [0,T]} \Bigl[ \frac{K^2}{M^2}|\widetilde F_{kk}(\xi,\zeta)| + \frac{T^2}{M_t^2}|\widetilde F_{ss}(\xi,\zeta)| \Bigr],
\end{equation}
where
\begin{equation}
|\widetilde F_{kk}(k,s)| \le \Bigl[ \frac{\|L\|^2T^2}{\pi(1+k^2)} + \frac{4k\|L\|T}{\pi(1+k^2)^2} + \frac{2(3k^2-1)}{\pi(1+k^2)^3} \Bigr] \cdot \|b\|
\end{equation}
as above, and 
\begin{equation}
|\widetilde F_{ss}(k,s)| \le \frac{\|b''\|}{\pi(1+k^2)} + \frac{2\|b'\|}{\pi(1+k^2)} (\|H\| + k\|L\|) + \frac{\|b\|}{\pi(1+k^2)} (\|H\| + k\|L\|)^2  + \frac{\|b\|}{\pi(1+k^2)} (\|H'\| + k\|L'\|).
\end{equation}

Thus, $\max_{\xi, \zeta} |\widetilde F_{kk}(\xi,\zeta)| = \Or(\|L\|^2T^2\|b\|)$, $\max_{\xi, \zeta} |\widetilde F_{ss}(\xi,\zeta)| = \Or((\|H\|_{C^1} + K\|L\|_{C^1})^2\|b\|_{C^2})$. Using $\|b\|_{L^1} \le \|b\|T$, to bound the error of the trapezoidal rule by $\epsilon$, we should choose
\begin{equation}
M = \Or\Bigl( \|L\|\|b\|_{L^1}^{3/2}\|b\|^{1/2}T^{3/2}/\epsilon^2 \Bigr) = \Or\Bigl( \|L\|\|b\|^2T^3/\epsilon^2 \Bigr) , 
\label{eqn:quadrature_imhom_k}
\end{equation}
and
\begin{equation}
M_t = \Or\Bigl( (\|H\|_{C^1}+\|L\|_{C^1}\|b\|_{L^1}/\epsilon)\|b\|_{L^1}^{1/2}\|b\|_{C^2}^{1/2}T^{3/2}/\epsilon \Bigr)
= \Or\Bigl( (\|H\|_{C^1}+\|L\|_{C^1})\|b\|_{C^2}^2T^3/\epsilon^2 \Bigr).
\label{eqn:quadrature_imhom_s}
\end{equation}

\REV{
\section{Introduction to Linear Combination of Unitaries}\label{app:LCU}
}
\REV{
Linear Combination of Unitaries (LCU)~\cite{ChildsWiebe2012,Kothari2014} is a quantum primitive that is widely employed for tasks such as Hamiltonian simulations~\cite{BerryChildsCleveEtAl2015,BerryChildsKothari2015} and quantum linear system algorithms (QLSAs)~\cite{ChildsKothariSomma2017}. LCU offers exponential enhancements in precision when compared to quantum phase estimation (QPE) methods. 
}

\REV{
Let $T = \sum_{i=0}^{K-1} \alpha_iU_i$ be a linear combination of unitaries $U_i$. Without loss of generality, we assume $K = 2^a$ and $\alpha_i>0$ since a phase factor can be subsumed into $U_i$. The \emph{selection oracle}
\begin{equation}
U = \sum_{i=0}^{K-1} \ket{i}\bra{i} \otimes U_i
\end{equation}
implements $U_i$ conditioned on the value of the controlled register. The \emph{preparation oracle} satisfies
\begin{equation}
V\ket{0^a} = \frac{1}{\|\alpha\|_1}\sum_{i=0}^{K-1} \sqrt{\alpha_i}\ket{i},
\end{equation}
where $\|\alpha\|_1 = \sum_i |\alpha_i|$. As shown in~\cite[Lemma 2.1]{Kothari2014}, we can implement $T$ in the following sense.
}

\REV{
\begin{lem}[LCU Lemma]\label{lem:LCU_Lemma}
    Let $W = (V^{\dagger} \otimes I_n)U(V \otimes I_n)$. Then for all states $\ket{\psi}$, 
    \begin{equation}
      W\ket{0^a}\ket{\psi} = \frac{1}{\|\alpha\|_1} \ket{0^a}T\ket{\psi} + \ket{\perp},
    \end{equation}
    where $\ket{\perp}$ is an unnormalized state that satisfies $(\ket{0^a}\bra{0^a} \otimes I_n)\ket{\perp} = 0$.
\end{lem}
}

\REV{
The LCU Lemma is a powerful quantum primitive, since the number of ancilla qubits $a$ only scales logarithmically in the number of linear combination terms $K$. Upon measuring the ancilla qubits of $W\ket{0^a}\ket{\psi}$ and obtaining the outcome $\ket{0^a}$, the resulting postselection state is  proportional to $T\ket{\psi}$.
This successful outcome occurs with probability $(\|T\ket{\psi}\|/\|\alpha\|_1)^2$. We can perform amplitude amplification~\cite{BrassardHoyerMoscaEtAl2002} to boost success probability to $\Omega(1)$, using $\Or(\|\alpha\|_1/\|T\ket{\psi}\|)$  number of queries to $U$, $V$, and $\ket{\psi}$.
}

\section{Hybrid implementation of the LCHS method}\label{app:LCHS_hybrid}

We discuss the evaluation of~\cref{eqn:observable} in a hybrid quantum-classical fashion. 
The idea is to use the quantum computer to evaluate $\braket{u_0|U^{\dagger}_k(t)OU_{k'}(t)|u_0}$ via the non-unitary Hadamard test and amplitude estimation, and then perform the summation via classical Monte Carlo sampling. 
In this section, we first discuss how to estimate $\braket{u_0|U^{\dagger}_k(t)OU_{k'}(t)|u_0}$ for a fixed pair $(k,k')$, and analyze the classical sampling complexity.

\subsection{Estimating observables via non-unitary Hadamard test}

\REV{
Hadamard test is a well-known quantum primitive to estimate the expectation value of a unitary matrix. 
It uses one ancilla qubit, applies a Hadamard gate on it, then applies a controlled version of the target unitary operation, followed by another Hadamard gate on the ancilla qubit. 
Then, the probability of obtaining $0$ when measuring the ancilla qubit is associated with the real part of the desired expectation value. 
The imaginary part can be obtained similarly with a slight modification. 
}

\REV{
The Hadamard test can be modified to estimate the expectation value $\braket{\phi|G|\phi}$ of a non-unitary matrix $G$ (see e.g., Ref.~\cite[Appendix D]{TongAnWiebe2021}).
The only modification is to replace the controlled unitary operation in the Hadamard test by the controlled version of the block encoding of $G$. For simplicity, we assume there is a unitary $U_G$ using $m$ extra ancilla qubits such that $(\bra{0}^{\otimes m} \otimes I)U_G (\ket{0}^{\otimes m} \otimes I) = \frac{1}{\alpha_G} G$, i.e., $U_G$ is a $(\alpha_G,0)$-block-encoding of $G$.  
Here $\alpha_G$ is the block-encoding factor with $\alpha_G \geq \|G\|$. 
Notice that when $G$ is unitary, the quantum circuit implementing $G$ itself is the $(1,0)$-block-encoding of $G$. 
Then, the scaled expectation value $\frac{1}{\alpha_G}\braket{\phi|G|\phi}$ can still be obtained via the the probability of obtaining $0$ when measuring the ancilla qubit. When combined with amplitude estimation, the non-unitary Hadamard test has the following complexity.
}

\REV{
\begin{lem}[{\cite[Lemma 7]{TongAnWiebe2021}}]
\label{lem:non_unitary_Hadamard}
    Suppose that $O_{\phi}$ is the state preparation oracle of $\ket{\phi}$, and $U_G$ is an $(\alpha_G,0)$-block-encoding of a matrix $G$. 
    Then, $\braket{\phi|G|\phi}$ can be estimated to precision $\epsilon$ with probability at least $1-\delta$, using $\mathcal{O}((\alpha_G/\epsilon) \log(\alpha_G/\epsilon)\log(1/\delta))$ queries to $O_{\phi}$, $U_G$ and their inverses. 
\end{lem}
}

\REV{
Using~\cref{lem:non_unitary_Hadamard}, we can directly obtain the complexity of estimating $\braket{u_0|U^{\dagger}_k(t)OU_{k'}(t)|u_0}$ in our case. 
}

\REV{
\begin{lem}\label{lem:sample_single}
    Suppose that $O_{\text{prep}}$ is the  state preparation oracle of $\ket{u_0}$, $U_O$ is an $(\alpha_O,0)$-block-encoding of $O$ with $\alpha_O \geq \norm{O}$, and $\widetilde{U}_k(t)$ is a quantum circuit that approximates $U_k(t)$ with error smaller than $\epsilon_{\text{HS}}$ for any $k$. 
    Then, $\braket{u_0|U^{\dagger}_k(t)OU_{k'}(t)|u_0}$ can be estimated to precision $\epsilon$ with probability at least $1-\delta$, by choosing $\epsilon_{\text{HS}} = \epsilon/(4\|O\|)$ and using $\mathcal{O}((\alpha_O/\epsilon) \log(\alpha_O/\epsilon)\log(1/\delta))$ queries to $O_{\text{prep}}$, $U_O$, $\widetilde{U}$ and their inverses. 
\end{lem}
}
\begin{proof}
    \REV{By~\cite[Lemma 53]{GilyenSuLowEtAl2019}, multiplying $\widetilde{U}_{k'}$, $U_O$ and $\widetilde{U}^{\dag}_{k}$ (with additional ancilla qubits) gives an $(\alpha_O,0)$-block-encoding of $\widetilde{U}^{\dagger}_k(t)O\widetilde{U}_{k'}(t)$, so we may directly use the non-unitary Hadamard test to estimate the desired expectation value. 
    According to~\cref{lem:non_unitary_Hadamard}, with probability at least $1-\delta$, we can estimate $\braket{u_0|\widetilde{U}^{\dagger}_k(t)O\widetilde{U}_{k'}(t)|u_0}$ to precision $\epsilon/2$ using $\mathcal{O}((\alpha_O/\epsilon) \log(\alpha_O/\epsilon)\log(1/\delta))$ queries. }

    \REV{There is another part of the error due to the imperfect implementation of $U_k$. 
    Specifically, we can bound 
    \begin{equation}
       \begin{split}
           & \quad \left|\braket{u_0|\widetilde{U}^{\dagger}_k(t)O\widetilde{U}_{k'}(t)|u_0} - \braket{u_0|U^{\dagger}_k(t)OU_{k'}(t)|u_0}\right| \\
           & \leq \norm{ \widetilde{U}^{\dagger}_k(t)O\widetilde{U}_{k'}(t) - U^{\dagger}_k(t)OU_{k'}(t) } \\
           & \leq 2 \|O\| \epsilon_{\text{HS}}, 
       \end{split}
    \end{equation}
    and thus we can choose $\epsilon_{\text{HS}} = \epsilon/(4\|O\|)$ to bound this part of the error by $\epsilon/2$. 
    }
\end{proof}

\REV{
\subsection{Linear combination via classical Monte Carlo sampling}
}
\REV{
Let $c = (c_k)$ denote the vector of the coefficients $c_k = \omega_j/(\pi(1+k_j^2))$. 
Notice that all $c_j$'s are positive real number and $\sum_{k,k'} c_k c_{k'} = \|c\|_1^2$.  
We describe a sampling-based hybrid algorithm for estimating~\cref{eqn:observable} as follows, which is similar to the sampling approaches from Fourier series  in~\cite{WangMcArdleBerta2023,Chakraborty2023}: 
\begin{enumerate}
    \item For each $j = 1,2,\cdots, J$, independently sample $(k,k')$ with probability $c_k c_{k'} / \|c\|_1^2$. 
    Denote the sample by $(k(j),k'(j))$. 
    \item Estimate $\braket{u_0|U^{\dagger}_{k(j)}(t)OU_{k'(j)}(t)|u_0}$ using the non-unitary Hadamard test. 
    Denote the successful estimator by $X_j$. 
    \item Estimate $u(t)^{*}Ou(t)$ using $\|c\|_1^2 \overline{X}$, where $\overline{X} = \frac{1}{J}\sum_{j=1}^J X_j$. 
\end{enumerate}
}

\REV{
If there is no error in implementing $U_k(t)$ or performing the non-unitary Hadamard test, then the the expectation value $\mathbb{E} \overline{X} $ is exactly  $u(t)^{*} O u(t)$. 
Taking into consideration these errors as well as the sampling errors, we can analyze the complexity of our approach as follows. 
}

\REV{
\begin{thm}\label{thm:LCHS_hybrid_complexity}
    Suppose that $O_{\text{prep}}$ is the  state preparation oracle of $\ket{u_0}$, $U_O$ is an $(\alpha_O,0)$-block-encoding of $O$ with $\alpha_O \geq \norm{O}$, and $\widetilde{U}_k(t)$ is a quantum circuit that approximates $U_k(t)$ with error smaller than $\epsilon_{\text{HS}}$ for any $k$. 
    Then, $u(t)^{*} O u(t)$ can be estimated to precision $\epsilon$ with probability at least $1-\delta$, by choosing $\epsilon_{\text{HS}} = \mathcal{O}(\epsilon/\norm{O})$. Furthermore, 
    \begin{enumerate}
        \item The number of the  samples is
        \begin{equation}
            \mathcal{O}\left( \frac{ \norm{O}^2 }{\epsilon^2} \log\left(\frac{1}{\delta}\right) \right), 
        \end{equation}
        \item Each circuit with sampling value $(k,k')$ uses   
        \begin{equation}
            \mathcal{O}\left(\frac{\alpha_O}{\epsilon} \log\left(\frac{\alpha_O}{\epsilon}\right)\log\left(\frac{\norm{O} \log(1/\delta)}{\delta \epsilon} \right)\right) 
        \end{equation}
        queries to $O_{\text{prep}}, U_O$ and $\widetilde{U}_k(t)$
    \end{enumerate}
\end{thm}
}
\begin{proof}
    \REV{Notice that there are two sources of errors and failure probabilities: the non-unitary Hadamard test and the sampling step. 
    It suffices to bound both errors (\emph{resp}. failure probabilities) by $\epsilon/2$ (\emph{resp}. $\delta/2$). }

    \REV{
    We first consider the sampling step. 
    Notice that, even in the presence of errors, each $X_j$ is bounded within $[-(\|O\|+1), \|O\|+1]$. 
    Hoeffding's inequality implies that 
    \begin{equation}
        \begin{split}
            \mathbb{P}\left( \abs{\norm{c}_1^2 \overline{X} - \norm{c}_1^2 \mathbb{E} \overline{X} } \geq \epsilon/2 \right) & = \mathbb{P}\left( \left|\sum_{j=1}^J X_j - \mathbb{E} \sum_{j=1}^J X_j \right| \geq \frac{J \epsilon}{2\norm{c}_1^2} \right) \\
            & \leq 2 \exp \left( - \frac{2 \left((J \epsilon)/(2\norm{c}_1^2)\right)^2 }{J (2(\|O\|+1))^2 } \right)  \\
            & = 2 \exp \left( - \frac{ J \epsilon^2 }{ 8 \norm{c}_1^4  (\|O\|+1)^2 } \right). 
        \end{split}
    \end{equation}
    To bound the failure probability by $\delta/2$, it suffices to choose 
    \begin{equation}
        J \geq \frac{8\norm{c}_1^4 (\norm{O}+1)^2 }{\epsilon^2} \log\left(\frac{4}{\delta}\right) = \mathcal{O}\left( \frac{\norm{c}_1^4 \norm{O}^2 }{\epsilon^2} \log\left(\frac{1}{\delta}\right) \right) = \mathcal{O}\left( \frac{ \norm{O}^2 }{\epsilon^2} \log\left(\frac{1}{\delta}\right) \right). 
    \end{equation}
    The last equality is because $c_j$'s come from the numerical approximation of the integral $\int \frac{\ud k}{\pi (1+k^2)}$ and $\norm{c}_1 = \sum_{j} \frac{\omega_j}{\pi(1+k_j^2)}\approx \int \frac{\ud k}{\pi (1+k^2)} = \mathcal{O}(1)$. 
    }

    \REV{
    Let $\delta'$, $\epsilon'$ denote the upper bound of the failure probability and the error in the non-unitary Hadamard test step for any $(k,k')$, respectively. 
    Recall that $X_j$ denotes the estimator upon success of the Hadamard test step.   
    Then 
    \begin{equation}
           \left| \norm{c}_1^2 \mathbb{E} \overline{X} - u(t)^{*} O u(t) \right| \leq \norm{c}_1^2 \sum_{k,k'} \frac{c_k c_{k'}}{\norm{c}_1^2} \epsilon' = \norm{c}_1^2 \epsilon'. 
    \end{equation}
    To bound this by $\epsilon/2$, it suffices to choose $\epsilon' = \epsilon/(2\norm{c}_1^2)$. 
    The overall success probability is at least $(1-\delta')^J \geq 1 - J\delta'$. 
    To bound it from below by $1 - \delta/2$, we can choose $\delta' = \delta/(2J)$. 
    By~\cref{lem:sample_single}, we need to choose the tolerated level of the error in implementing $U_k(t)$ as $\epsilon_{\text{HS}} = \epsilon'/(4\norm{O})$, and the query complexity for each circuit is 
    \begin{equation}
        \mathcal{O}\left(\frac{\alpha_O}{\epsilon'} \log\left(\frac{\alpha_O}{\epsilon'}\right)\log\left(\frac{1}{\delta'}\right)\right). 
    \end{equation}
    Plugging in the choices of $\epsilon'$, $\delta'$ and $J$ yields the desired complexity estimate. 
    }
\end{proof}

\REV{In~\cref{thm:LCHS_hybrid_complexity}, the query complexity of each circuit is measured by the number of queries to the state preparation oracle, the block-encoding of $O$ and the circuit $\widetilde{U}_k(t)$ that approximates $U_k(t)$. 
The final query complexity in terms of the input models for $A(t)$ still depends on how we actually implement $\widetilde{U}_k(t)$. 
This may introduce extra overhead if we use low-order methods such as first-order Trotter formula. 
On the other hand, if we use high-order Trotter formula or truncated Dyson series method, the complexity of implementing $\widetilde{U}_k(t)$ can be almost linear in $t$, $k$ and certain norms of $A$, and the dependence on $\epsilon_{\text{HS}}$ is $1/\epsilon_{\text{HS}}^{o(1)}$. 
In the worst case, the query complexity of the circuit still has an almost linear dependence on $K$. 
However, we remark that in the average case, such a $K$ dependence can be improved. 
This is because the circuit with larger $k$ and $k'$ is sampled with smaller probability, which decays quadratically as $\sim 1/(k^2 k'^2)$. 
}

\REV{Throughout this section, we care about the expectation value $u(t)^{*} O u(t)$ associated with the possibly unnormalized solution $u(t)$, which is the natural setup in various applications of classical non-unitary dynamics. 
If we want to estimate $\braket{u(t)|O|u(t)}$ where $\ket{u(t)}$ is the normalized solution of the ODE, then we need to divide the previous estimator by $\norm{u(t)}^2$. 
This may affect the algorithm in two aspects. 
First, to ensure that the final estimator is still an $\epsilon$-approximation, we need to change the tolerated level of errors in $u(t)^{*} O u(t)$ to be $\epsilon\norm{u(t)}^2$. 
Second, if $\norm{u(t)}$ is unknown \emph{a priori}, then an extra algorithm to estimate it is required. 
In this case, one may consider a similar hybrid algorithm specified in e.g., Ref.~\cite{WangMcArdleBerta2023}. 
}

\section{Implementation of the inhomogeneous term}\label{app:inhomo_implementation}

The inhomogeneous term in~\cref{eqn:Duhamel} can be implemented following the same approach as the homogeneous case discussed in the main text. 
After discretizing the integral for both $k$ and $s$ using another trapezoidal rule, we obtain $\int_0^T \mathcal{T} e^{-\int_s^t A(s') \ud s'} b(s) \ud s \approx \sum_{j'=0}^{M_t} \sum_{j=0}^M  \widetilde{c}_{j,j'} \mathcal{T} e^{- i \int_{s_{j'}}^t (H(s')+k_j L(s')) \ud s' } \ket{b(s_{j'})}$. 
A further Trotterization yields 
\begin{equation}\label{eqn:Trotter_inhomo_td}
\begin{split}
    & \quad \int_0^T \mathcal{T} e^{-\int_s^t A(s') \ud s'} b(s) \ud s \\
    & \approx \sum_{j'=0}^{M_t} \sum_{j=0}^M  \widetilde{c}_{j,j'} \prod_{l'=0}^{r-1} \prod_{l=0}^{\Xi_p-1} \left(e^{-i H(s_{j'}+(l'+ \delta_l) (t-s_{j'})/r)\beta_l (t-s_{j'})/r} e^{-i L (s_{j'}+(l'+ \gamma_l) (t-s_{j'})/r) \alpha_l k_j(t-s_{j'})/r}\right) \ket{b(s_{j'})}. 
\end{split}
\end{equation}
Suppose we are given the coefficient oracle 
$O_{\text{coef}}': \ket{0} \rightarrow \frac{1}{\sqrt{\|\widetilde{c}\|_1}}\sum_{j'=0}^{M_t} \sum_{j=0}^{M} \sqrt{\widetilde{c}_{j,j'}} \ket{j'}\ket{j}$, 
the matrix input oracles $O_L'(s,\tau_0,\tau_1) = \sum_{j'=0}^{M_t} \ket{j'}\bra{j'}\otimes  e^{ -iL(\tau_0 j' + \tau_1 (M_t-j'))s(M_t-j')}$ for $|s|M_t \leq 1/\|L\|$,  
$O_H'(s,\tau_0,\tau_1) = \sum_{j'=0}^{M_t} \ket{j'}\bra{j'} \otimes e^{-iH(\tau_0 j' + \tau_1 (M_t-j'))s(M_t-j')}$ for $|s|M_t \leq 1/\|H\|$, 
and the source term input oracle $O_{b}: \ket{j'}\ket{0} \rightarrow \ket{j'}\ket{b(s_{j'})}$. 
Notice that the input models of $H(s)$, $L(s)$ and $b(s)$ are given in a coherent manner, which is an extension of the time-dependent matrix encoding proposed in~\cite{LowWiebe2019}. 
According to the binary representation of $j$'s, we may construct the LCU select oracle $\text{SEL}_L'(s,\tau_0,\tau_1) = \sum_{j'=0}^{M_t} \sum_{j=0}^M \ket{j'}\bra{j'} \otimes \ket{j}\bra{j} \otimes e^{-iL(\tau_0 j' + \tau_1 (M_t-j'))k_j  s(M_t-j')}$ with logarithmic cost. 
The main steps of the LCU are as follows. 
Starting with all-zero state $\ket{0}_{a'}\ket{0}_a\ket{0}$ where two ancilla registers are for encoding the time and space indices, we first apply $O_{\text{coef}}'$ to create the superposition in the ancilla registers, then apply $O_{b}$ to encode the information of $b(t)$ into the output register, then apply $\text{SEL}'_L(\alpha_lt/(rM_t),t/M_t,(l'+\gamma_l)t/(rM_t))$ and $O'_H(\beta_lt/(rM_t),t/M_t,(l'+\delta_l)t/(rM_t))$ for all $l\in[\Xi_p],l'\in [r]$, and finally apply $O_{\text{coef}}'^{\dagger}$. 
The resulting state approximately encodes $\int_0^T e^{-A(T-s)} b(s) \ud s$ in its first subspace with all ancilla qubits to be $0$.

\section{Proofs of the complexity estimates}\label{app:proofs}

Here we present the detailed proofs of the complexity estimates of our algorithm. 
This section is organized as follows. 
We first discuss how to construct the select oracles used in the LCU procedure with logarithmic cost. 
Then we analyze the complexity of encoding the homogeneous term and the inhomogeneous term, followed by a result on linearly combining two general quantum states. 
The proof of our main result (\cref{thm:td_inhomo}) is a direct consequence of all the above results, which is presented at the end of this section. 

\subsection{Construction of the select oracles}\label{app:proofs_oracles}

In the LCU approach, we need the select oracle $\sum_{j=0}^{J-1} \ket{j}\bra{j} \otimes U_j$ for the unitaries $U_j$ to be combined. 
This can be constructed with only $\log(J)$ cost if each $U_j$ can be efficiently implemented and different $U_j$'s are related in certain ways. 
Such a technique has been widely used in e.g.~\cite{ChildsKothariSomma2017,LowWiebe2019,AnFangLin2022}, and for completeness we state the result in the following lemma. 

\begin{lem}\label{lem:select_oracle_binary}
    Let $U_j$ be a set of unitaries such that $U_j = U_0^j$. 
    Then the select oracle $\sum_{j=0}^{J-1} \ket{j}\bra{j} \otimes U_j$ can be constructed with $\lceil\log(J)\rceil$ queries to controlled $U_j$'s. 
\end{lem}
\begin{proof}
     For all $0 \leq l \leq \lceil\log(J)\rceil$, we perform c-$U_{2^l}$ on the system register controlled by the $l$-th qubit of the first register if it is $1$. 
     Due to the binary representation of $j$, if the first register is $\ket{j}$, then the operator on the system register is actually $U_j = U_0^j$. 
     Notice that when applying $U_{2^l}$, we directly perform the circuit for $U_{2^l}$ rather than applying $U_0$ for $2^{l}$ times. 
\end{proof}

Suppose that we are given an oracle $O_L(s,\tau) = e^{-iL(\tau) s}$. 
Notice that in the trapezoidal rule, we have $k_j = -K+2jK/M$. 
We write the select oracle as 
\begin{equation}
    \begin{split}
        \text{SEL}_L (s,\tau) &= \sum_{j=0}^M \ket{j}\bra{j} \otimes e^{-iL (\tau ) k_j s} \\
        & = \sum_{j=0}^M \ket{j}\bra{j} \otimes e^{-iL (\tau ) (-K+2jK/M) s} \\
        & = e^{iL (\tau ) Ks }\sum_{j=0}^M \ket{j}\bra{j} \otimes \left(e^{-iL (\tau )2sK/M  }\right)^j. 
    \end{split}
\end{equation}
The operator $e^{iL (\tau ) Ks }$ can be directly implemented by $O_L$ by noting that in our algorithm all the $s$'s will be chosen sufficiently small such that $K|s| < 1/\|L\|$, and the second operator can be constructed by~\cref{lem:select_oracle_binary} with cost $\mathcal{O}(\log(M))$. 

The select oracle $\text{SEL}_L'(s,\tau_0,\tau_1)$ can be constructed via exactly the same approach, using $\mathcal{O}(\log(M))$ queries to $O_L'(s,\tau)$. 
Therefore in our complexity analysis, we will directly estimate the number of queries to the select oracles, and the overall query complexity will be multiplied by $\log(M)$ at the end.

\subsection{Homogeneous term}

We start with the homogeneous case. 
Let 
\begin{equation}
    \Gamma_p = \max_{0\leq q \leq p, \tau \in [0,T]} \left( \|H^{(q)}(\tau)\|+ \|L^{(q)}(\tau)\|\right)^{1/(q+1)}. 
\end{equation}

\begin{lem}\label{lem:BE_td_homo}
    There exists a quantum algorithm which gives a  $(\|c\|_1,\log(M),\epsilon)$-block-encoding of $\mathcal{T}e^{-\int_0^T A(s) \ud s}$, using $\mathcal{O}(1)$ queries to $O_{\text{coef}}$ and 
    \begin{equation}
        \mathcal{O} \left( \Gamma_p^{1+1/p} \frac{\|c\|_1^{1/p} T^{1+1/p}}{\epsilon^{1+2/p}} \log\left(\frac{\|L\|T}{\epsilon}\right) \right)
    \end{equation}
    queries to $O_L(s,\tau)$ and $O_H(s,\tau)$. 
\end{lem}

\begin{proof}
    According to~\cite[Lemma 52]{GilyenSuLowEtAl2019}, the LCU yields a $(\|c\|_1,\log(M),0)$-block-encoding of 
    \begin{equation}
        \sum_{j=0}^M c_j \prod_{l'=0}^{r-1} \prod_{l=0}^{\Xi_p-1} \left(e^{-i H((l'+ \delta_l) T/r)\beta_l T/r} e^{-iL ((l'+ \gamma_l) T/r) \alpha_l k_jT/r}\right), 
    \end{equation}
    which is in turn a $(\|c\|_1,\log(M),\epsilon)$-block-encoding of $\mathcal{T}e^{-\int_0^T A(s) \ud s}$ if 
    \begin{equation}
        \left\|  \sum_{j=0}^M c_j \prod_{l'=0}^{r-1} \prod_{l=0}^{\Xi_p-1} \left(e^{-i H((l'+ \delta_l) T/r)\beta_l T/r} e^{-iL ((l'+ \gamma_l) T/r) \alpha_l k_jT/r}\right) - \mathcal{T}e^{-\int_0^T A(s) \ud s} \right\| \leq \epsilon. 
    \end{equation}
    We can guarantee this by choosing sufficiently large $K$, $M$ and $r$. 

    First, based on \cref{eqn:quadrature_homogeneous}, we can choose 
    \begin{equation}
        M = \Or\left(\frac{\|L\|T}{\epsilon^2}\right), \quad K = \mathcal{O}\left(\frac{1}{\epsilon}\right) 
    \end{equation}
    so that 
    \begin{equation}
        \left\|\mathcal{T}e^{-\int_0^T A(s) \ud s} - \sum_{j=0}^M c_j \mathcal{T}e^{-i\int_0^T (H(s)+k_jL(s)) \ud s}\right\| \leq \frac{\epsilon}{2}. 
    \end{equation}
    Next, according to~\cite{WiebeBerryHoyerEtAl2010}, we have 
    \begin{equation}\label{eqn:Trotter_error_td}
    \begin{split}
        & \quad \norm{ \prod_{l'=0}^{r-1} \prod_{l=0}^{\Xi_p-1} \left(e^{-i H((l'+ \delta_l) T/r)\beta_l T/r} e^{-iL ((l'+ \gamma_l) T/r) \alpha_l k_jT/r}\right) - \mathcal{T}e^{-i\int_0^T (H(s)+k_jL(s)) \ud s}} \\
        &\leq \mathcal{O}\left( \left(\max_{0\leq q \leq p, \tau \in [0,T]} \left( \|H^{(q)}(\tau)\|+ k_j \|L^{(q)}(\tau)\|\right)^{1/(q+1)}\right)^{p+1} \frac{T^{p+1}}{r^p}\right) \\
        & \leq \mathcal{O}\left( \Gamma_p^{p+1} \frac{K^{p+1} T^{p+1}}{r^p}\right), 
    \end{split}
    \end{equation}
    and 
    \begin{equation}
    \begin{split}
        & \quad \left\| \sum_{j=0}^M c_j \mathcal{T}e^{-i\int_0^T (H(s)+k_jL(s)) \ud s} - \sum_{j=0}^M c_j \prod_{l'=0}^{r-1} \prod_{l=0}^{\Xi_p-1} \left(e^{-i H((l'+ \delta_l) T/r)\beta_l T/r} e^{-iL ((l'+ \gamma_l) T/r) \alpha_l k_jT/r}\right)\right\| \\
        & \leq \|c\|_1 \max_j \norm{ \prod_{l'=0}^{r-1} \prod_{l=0}^{\Xi_p-1} \left(e^{-i H((l'+ \delta_l) T/r)\beta_l T/r} e^{-iL ((l'+ \gamma_l) T/r) \alpha_l k_jT/r}\right) - \mathcal{T}e^{-i\int_0^T (H(s)+k_jL(s)) \ud s}} \\
        & \leq \mathcal{O} \left(  \Gamma_p^{p+1} \frac{\|c\|_1 K^{p+1} T^{p+1}}{r^p}  \right). 
    \end{split}
    \end{equation}
    To bound this by $\epsilon/2$, it suffices to choose 
    \begin{equation}
    \begin{split}
        r &= \mathcal{O} \left( \Gamma_p^{1+1/p} \frac{\|c\|_1^{1/p} K^{1+1/p} T^{1+1/p}}{\epsilon^{1/p}} \right) \\
        & = \mathcal{O} \left( \Gamma_p^{1+1/p} \frac{\|c\|_1^{1/p} T^{1+1/p}}{\epsilon^{1+2/p}} \right)
    \end{split}
    \end{equation}

    Query complexity in $O_{\text{coef}}$ is straightforward, and that in the select oracles are $\mathcal{O}(r)$. 
    Using our estimate of $r$ (and noticing that the select oracle can be constructed with $\mathcal{O}(\log(M))$ cost), we complete the proof. 
\end{proof}

\begin{thm}\label{thm:td_homo}
    Consider the homogeneous ODE~\cref{eqn:inhom_general_diff_eq} with $b(t) \equiv 0$. 
    Then, there exists a quantum algorithm that prepares an $\epsilon$-approximation of the state $\ket{\mathcal{T}e^{-\int_0^T A(s) \ud s}u_0}$ with $\Omega(1)$ success probability and a flag indicating success, using queries to $O_L(s,\tau)$ and $O_H(s,\tau)$ a total number of times 
    \begin{equation}
        \mathcal{O} \left( \Gamma_p^{1+1/p} \left(\frac{\|u_0\|}{\|u(T)\|}\right)^{2+2/p} \frac{T^{1+1/p}}{\epsilon^{1+2/p}} \log\left( \frac{\|u_0\|\|L\|T}{\|u(T)\|\epsilon}\right) \right), 
    \end{equation}
    queries to $O_{\text{coef}}$ and $O_{\text{prep}}$ for $\mathcal{O}\left(\frac{\|u_0\|}{\|u(T)\|}\right)$ times, and $\mathcal{O}\left(\log\left(\frac{\|u_0\|\|L\|T}{\|u(T)\|\epsilon}\right)\right)$ ancilla qubits. 
\end{thm}

\begin{proof}
    Let $V$ denote the $(\|c\|_1,\log(M),\epsilon')$-block-encoding of $\mc{T} e^{-\int_0^T A(s) \ud s}$ as constructed in~\cref{lem:BE_td_homo}. 
    We further write 
    \begin{equation}
        \|c\|_1 \bra{0}_a V \ket{0}_a = \mc{T} e^{-\int_0^T A(s) \ud s} + E
    \end{equation}
    where $\|E\| \leq \epsilon'$. 
    We start with the state $\ket{0}_a\ket{0}$, where the ancilla register contains $\log(M)$ qubits. 
    After applying $O_{\text{prep}}$ on the system register and $V$, we obtain the state 
    \begin{equation}
        V (I_a \otimes O_{\text{prep}}) \ket{0}_a\ket{0} = \frac{1}{\|c\|_1} \ket{0}_a (\mc{T} e^{-\int_0^T A(s) \ud s})\ket{u_0} + \ket{\perp}. 
    \end{equation}

    Using the inequality $\|x/\|x\|-y/\|y\|\|\leq 2\|x-y\|/\|x\|$ for two vectors $x,y$, we can bound the error in the quantum state after a successful measurement as 
    \begin{equation}
        \left\|\ket{(\mc{T} e^{-\int_0^T A(s) \ud s}+E)u_0} - \ket{u(T)}\right\| \leq \frac{2\|Eu_0\|}{\|u(T)\|} \leq \frac{2\epsilon'\|u_0\|}{\|u(T)\|}. 
    \end{equation}
    In order to bound this error by $\epsilon$, it suffices to choose 
    \begin{equation}
        \epsilon' = \frac{\epsilon \|u(T)\|}{2\|u_0\|}. 
    \end{equation}
    Therefore the overall complexity of a single run of our algorithm can be obtained by~\cref{lem:BE_td_homo} with $\epsilon'$, which becomes 
    \begin{equation}
        \mathcal{O} \left( \Gamma_p^{1+1/p} \left(\frac{\|u_0\|}{\|u(T)\|}\right)^{1+2/p} \frac{\|c\|_1^{1/p} T^{1+1/p}}{\epsilon^{1+2/p}} \log\left(\frac{\|u_0\|\|L\|T}{\|u(T)\|\epsilon}\right) \right)
    \end{equation}
    queries to the matrix input, and $\mathcal{O}(1)$ to the state preparation oracle. 

    The expected number of repeats to get a success, after amplitude amplification, is
    \begin{equation}
        \mathcal{O}\left(\frac{\|c\|_1\|u_0\|}{\|(\mc{T} e^{-\int_0^T A(s) \ud s}+E)u_0\| }\right) \leq \mathcal{O}\left(\frac{\|c\|_1\|u_0\|}{\|u(T)\| - \epsilon' \|u_0\| }\right) = \mathcal{O}\left(\frac{\|c\|_1\|u_0\|}{\|u(T)\|} \right). 
    \end{equation}
    Since $\|c\|_1$ is the trapezoidal rule of the integral $\int_{-K}^K \frac{\ud k}{\pi (1+k^2)}$, we have $\|c\|_1 = \mathcal{O}(1)$. 
    This completes the proof of~\cref{thm:td_homo}. 
\end{proof}

\subsection{Inhomogeneous term}

Here we analyze the complexity of encoding the inhomogeneous term $\int_0^T \mc{T}  e^{ -\int_s^T A(s')\ud s'} b(s) \ud s$. 

\begin{lem}\label{lem:td_inhomo}
    There exists a quantum algorithm which maps $\ket{0}_{a'}\ket{0}_a\ket{0}$ to the state $\frac{1}{\eta} \ket{0}_{a'}\ket{0}_a \widetilde{u} + \ket{\perp}$ such that $\widetilde{u}$ is an $\epsilon$-approximation of $\int_0^T \mc{T}  e^{ -\int_s^T A(s')\ud s'} b(s) \ud s$ and $\eta = \norm{\widetilde{c}}_1$, using queries to the input models of $H$ and $L$ a total number of times 
    \begin{equation}
    \mathcal{O}\left( \Gamma_p^{1+1/p} \frac{\|\widetilde{c}\|_1^{1/p} \|b\|_{L^1}^{1+1/p} T^{1+1/p}}{\epsilon^{1+2/p}} \log\left(\frac{\|L\|\|b\|T}{\epsilon}\right) \right). 
    \end{equation}
    queries to $O_{\text{coef}}'$ and $O_b$ for $\mathcal{O}(1)$ times, and $\log(\Gamma_1 \|b\|_{C^2}T/\epsilon)$ ancilla qubits. 
\end{lem}
\begin{proof}
    According to~\cref{app:inhomo_implementation}, the LCU procedure before measurement gives the state $\frac{1}{\eta} \ket{0}_{a'}\ket{0}_a \widetilde{u} + \ket{\perp}$ such that $\eta = \|\widetilde{c}\|_1$ and 
    \begin{equation}
    \begin{split}
        \widetilde{u} &= \sum_{j'=0}^{M_t} \sum_{j=0}^M  \widetilde{c}_{j,j'} \prod_{l'=0}^{r-1} \prod_{l=0}^{\Xi_p-1} \left(e^{-i H(s_{j'}+(l'+ \delta_l) (T-s_{j'})/r)\beta_l (T-s_{j'})/r} e^{-i L (s_{j'}+(l'+ \gamma_l) (T-s_{j'})/r) \alpha_l k_j(T-s_{j'})/r}\right) \ket{b(s_{j'})}. 
    \end{split}
    \end{equation}
    To bound the error by $\epsilon$, we use \cref{eqn:general_A_inhom_truncate} and~\cite{WiebeBerryHoyerEtAl2010} to obtain 
    \begin{equation}
    \begin{split}
        & \quad \norm{\widetilde{u} - \int_0^T \mathcal{T} e^{-\int_s^T A(s') \ud s'} b(s) \ud s } \\
        & \leq \norm{ \int_0^T \mathcal{T} e^{-\int_s^T A(s') \ud s'} b(s) \ud s  - \sum_{j'=0}^{M_t} \sum_{j=0}^M  \widetilde{c}_{j,j'} \mathcal{T} e^{- i \int_s^T (H(s')+k_jL(s') ) \ud s'} \ket{b(s_{j'})}  } \\
        & \quad + \norm{ \sum_{j'=0}^{M_t} \sum_{j=0}^M  \widetilde{c}_{j,j'} \mathcal{T} e^{- i \int_s^T (H(s')+k_jL(s') ) \ud s'} \ket{b(s_{j'})}  - \widetilde{u} }. 
    \end{split}
    \end{equation}
    Therefore, to bound the error by $\epsilon$, it suffices to choose~\cref{app:quadrature}
    \begin{equation}
        K = \mathcal{O}\left( \frac{\|b\|_{L^1}}{\epsilon}\right), \quad M = \Or\left(\text{poly}\left(\frac{\|L\|\|b\|T}{\epsilon}\right)\right), \quad M_t = \mathcal{O}\left( \text{poly}\left(\frac{\Gamma_1 \|b\|_{C^2} T}{\epsilon}\right)\right), 
    \end{equation}
    and 
    \begin{equation}
    r = \mathcal{O}\left( \Gamma_p^{1+1/p} \frac{\|\widetilde{c}\|_1^{1/p} K^{1+1/p} T^{1+1/p}}{\epsilon^{1/p}} \right). 
    \end{equation}
    The overall query complexity to the matrix input models is $\mathcal{O}(r\log(M))$. 
\end{proof}

\subsection{Linear combination of homogeneous and inhomogeneous terms}

We state a more general result as follows on how to linearly combine two (``block-encoded'') quantum states. 
The algorithm is a special case of LCU, but the errors in the quantum states need to be carefully controlled. 

\begin{lem}\label{lem:LCS}
    Let $x_0$ and $x_1$ denote two (possibly unnormalized) vectors. 
    Suppose that we are given two unitaries $U_0$ and $U_1$ such that $U_j \ket{0}_a\ket{0} = \frac{1}{\eta_j} \ket{0}_a \widetilde{x}_j + \ket{\perp}$ where $\norm{ x_j - \widetilde{x}_j } \leq \epsilon_j$. 
    Then, for any real positive parameters $(\theta_0,\theta_1)$, there exists a quantum algorithm which outputs an $\epsilon$-approximation of the quantum state $\ket{\theta_0 x_0 + \theta_1 x_1}$ with $\Omega(1)$ success probability and a flag indicating success, using $1$ extra ancilla qubit and $\mathcal{O}\left( \frac{\eta_0\theta_0+\eta_1\theta_1}{\norm{\theta_0 x_0 + \theta_1 x_1} } \right)$ queries to $U_0$, $U_1$ and additional one-qubit gate, where the tolerated errors are chosen as 
    \begin{equation}
        \epsilon_0 = \frac{\norm{\theta_0 x_0 + \theta_1 x_1}\epsilon}{4\theta_0}, \quad \epsilon_1 = \frac{\norm{\theta_0 x_0 + \theta_1 x_1}\epsilon}{4\theta_1}. 
    \end{equation}
\end{lem}
\begin{proof}
    Let $R$ be a single-qubit rotation such that 
    \begin{equation}
        R\ket{0} = \frac{1}{\sqrt{\eta_0 \theta_0 + \eta_1 \theta_1}} \left( \sqrt{\eta_0 \theta_0} \ket{0} + \sqrt{\eta_1 \theta_1} \ket{1} \right). 
    \end{equation}
    Then 
    \begin{equation}
        \begin{split}
            & \quad (R^{\dagger} \otimes I \otimes I)\left(\ket{0}\bra{0}\otimes U_0 + \ket{0}\bra{0}\otimes U_1 \right)(R\otimes I \otimes I) \ket{0}_c\ket{0}_a\ket{0} \\
            & = \frac{1}{\eta_0\theta_0+\eta_1\theta_1} \ket{0}_c\ket{0}_a (\theta_0 \widetilde{x}_0 + \theta_1 \widetilde{x}_1) + \ket{\perp}. 
        \end{split}
    \end{equation}
    In order to bound the error in the quantum state by $\epsilon$, it suffices to bound $\norm{(\theta_0 \widetilde{x}_0 + \theta_1 \widetilde{x}_1) - (\theta_0 x_0 + \theta_1 x_1)} $ by $\norm{\theta_0 x_0 + \theta_1 x_1}\epsilon/2$, so we may choose 
    \begin{equation}
        \epsilon_0 = \frac{\norm{\theta_0 x_0 + \theta_1 x_1}\epsilon}{4\theta_0}, \quad \epsilon_1 = \frac{\norm{\theta_0 x_0 + \theta_1 x_1}\epsilon}{4\theta_1}. 
    \end{equation}
    Since the complexity of each run is $\mathcal{O}(1)$, the overall complexity is the number of repeats to get a success, which, after amplitude amplification, becomes 
    \begin{equation}
        \mathcal{O}\left( \frac{\eta_0\theta_0+\eta_1\theta_1}{\norm{\theta_0 \widetilde{x}_0 + \theta_1 \widetilde{x}_1} } \right) = \mathcal{O}\left( \frac{\eta_0\theta_0+\eta_1\theta_1}{\norm{\theta_0 x_0 + \theta_1 x_1} } \right). 
    \end{equation}
\end{proof}

\subsection{Proof of \texorpdfstring{\cref{thm:td_inhomo}}{} }

\begin{proof}[Proof of~\cref{thm:td_inhomo}]
    According to~\cref{lem:BE_td_homo}, we can apply $O_{\text{prep}}$ and the $(\norm{c}_1,\log(M),\epsilon_0)$-block-encoding of $\mathcal{T} e^{-\int_0^T A(s) \ud s}$ to obtain the state 
    \begin{equation}
        \frac{1}{\norm{c}_1 \|u_0\|} \ket{0}_a \widetilde{u}_0 + \ket{\perp}, 
    \end{equation}
    where $\norm{\widetilde{u}_0 - \mathcal{T} e^{-\int_0^T A(s) \ud s} u_0 } \leq \epsilon_0 \|u_0\|$. 
    The algorithm in~\cref{lem:td_inhomo} gives the state 
    \begin{equation}
        \frac{1}{\norm{\widetilde{c}}_1} \ket{0}_{a'}\ket{0}_a \widetilde{u}_1 + \ket{\perp}, 
    \end{equation}
    where $\norm{\widetilde{u}_1 - \int_0^T \mathcal{T}e^{-\int_s^T A(s') \ud s'} b(s) \ud s} \leq \epsilon_1$. 
    So an $\epsilon$-approximation of $\ket{u(T)}$ can be directly constructed using~\cref{lem:LCS} with $\theta_0=\theta_1 = 1$, $\eta_0 = \norm{c}_1 \|u_0\|$, $\eta_1 = \norm{\widetilde{c}}_1$, and choose 
    \begin{equation}
        \epsilon_0 = \frac{\norm{u(T)}\epsilon}{4\|u_0\|}, \quad \epsilon_1 = \frac{\norm{u(T)}\epsilon}{4}. 
    \end{equation}
    
    The overall complexity can be estimated using~\cref{lem:BE_td_homo} and~\cref{lem:td_inhomo}. 
    Specifically, notice that $\|c\|_1 = \mathcal{O}(1)$ as bounded in the proof of~\cref{thm:td_homo} and 
    $\|\widetilde{c}\|_1 = \mathcal{O}(\|b\|_{L^1})$ since $\|\widetilde{c}\|_1$ is the discretized integral $\int_0^T \int_{\RR} \frac{1}{\pi (1+k^2)} |b(s)| \ud k \ud s $ via trapezoidal rule. 
    The number of queries to the matrix input oracles becomes 
    \begin{equation}
        \begin{split}
            & \quad \mathcal{O}\left(  \frac{\eta_0\theta_0+\eta_1\theta_1}{\norm{\theta_0 x_0 + \theta_1 x_1} } \Gamma_p^{1+1/p} \left( \frac{\|c\|_1^{1/p} T^{1+1/p}}{\epsilon_0^{1+2/p}} + \frac{\|\widetilde{c}\|_1^{1/p} \|b\|_{L^1}^{1+1/p} T^{1+1/p}}{\epsilon_1^{1+2/p}} \right) \log(M) \right)\\
            &= 
            \mathcal{O}\left( \left( \frac{\|u_0\|+\|b\|_{L^1}}{\norm{u(T)} } \right)^{2+2/p} \Gamma_p^{1+1/p} \frac{  T^{1+1/p}}{\epsilon^{1+2/p}}  \log\left(\frac{\|u_0\|+\|b\|}{\|u(T)\|} \frac{\|L\|T}{\epsilon}\right) \right),  
        \end{split}
    \end{equation}
    and the number of queries to $O_{\text{prep}}$, $O_b$ and additional one-qubit gate is $\mathcal{O}\left( \frac{\eta_0\theta_0+\eta_1\theta_1}{\norm{\theta_0 x_0 + \theta_1 x_1} } \right) = \mathcal{O}\left( \frac{\|u_0\|+\|b\|_{L^1}}{\norm{u(T)} } \right)$. 
\end{proof}

\REV{\cref{thm:td_inhomo} requires $A(t)$ to be $p$-th order continuously differentiable due to the dependence on $\Gamma_p$. 
This assumption is solely due to the usage of the $p$-th order product formula and may be weakened using other methods for solving the time-dependent Hamiltonian simulation problems. 
For example,  the LCHS with the truncated Dyson series method~\cite{LowWiebe2019} only requires $A(t)$ to be first-order continuously differentiable to achieve high order accuracy. }

\section{ODEs with time-independent matrix \texorpdfstring{$A$}{} }

When $A(t) \equiv A$ is time-independent, we will only assume the Hamiltonian simulation oracles $O_L(s) = e^{-i L s}$ and $O_H(s) = e^{-i H s}$ for fixed $s$. 
Then all the coherent encoding of the time evolution and the select oracles in LCU procedure can be constructed using $\mathcal{O}(\log(M)\log(M_t))$ queries to $O_L$ and $O_H$. 
In particular, $\text{SEL}_L$ can be constructed in the same way as in~\cref{app:proofs_oracles}. 
The oracle 
\begin{equation}
    O_L' = \sum_{j'=0}^{M_t} \ket{j'}\bra{j'} \otimes e^{-i L s (M_t-j')} = \sum_{j'=0}^{M_t} \ket{M_t-j'}\bra{M_t-j'} \otimes \left( {e^{-i L s}}\right)^{j'} , 
\end{equation}
which can be implemented by applying Pauli-X gates on each qubit of the ancilla register then using~\cref{lem:select_oracle_binary} with $\mathcal{O}(\log(M_t))$ queries to $O_L$. 
The oracle 
\begin{equation}
    \text{SEL}_L' = \sum_{j'=0}^{M_t} \sum_{j=0}^M \ket{j'}\bra{j'} \otimes \ket{j} \bra{j} \otimes e^{- i L k_j s (M_t-j')} = \sum_{j=0}^M   \ket{j} \bra{j} \otimes \left( \sum_{j'=0}^{M_t} \ket{j'}\bra{j'} \otimes  e^{- i L s (M_t-j')}  \right)^{k_j}. 
\end{equation}
So it can be constructed using $O_L'$ for $\mathcal{O}(\log(M))$ times and thus $O_L$ for $\mathcal{O}(\log(M)\log(M_t))$ times. 
The oracles for $H$ can be constructed in a similar (and even simpler) manner. 

The algorithm described in the main text can be directly applied to this special case, and the product formula degenerates to its time-independent version. 
For the time-independent product formula, we may use a better error bound proved in~\cite{ChildsSuTranEtAl2020} that the short time error (\emph{i.e.},~\cref{eqn:Trotter_error_td}) now becomes 
\begin{equation}
    \mathcal{O}\left( \sum_{H_q \in\left\{H,k_j L \right\} }\norm{[H_p,\cdots,[H_1,H_0]]} \frac{T^{p+1}}{r^p}  \right) = \mathcal{O} \left( \Lambda_p^{p+1} \frac{K^p T^{p+1}}{r^p}  \right), 
\end{equation}
where 
\begin{equation}
    \Lambda_p = \left(\sum_{H_q \in\left\{H, L \right\} }\norm{[H_p,\cdots,[H_1,H_0]]}\right)^{1/(p+1)}. 
\end{equation}
Compared to~\cref{eqn:Trotter_error_td}, we can replace the parameter $\Gamma_p$ by $\Lambda_p$ and reduce the order of $K$ by $1$ in the Trotter error. 
This will yield an improved complexity estimates since $\Gamma_p \geq \Lambda_p$. 

For completeness, we state the complexity of our algorithm applied to ODEs with time-independent $A$ in the following two theorems.

\begin{thm}\label{thm:ti_homo}
    Consider the ODE~\cref{eqn:inhom_general_diff_eq} with time-independent $A(t)\equiv A$ and $b(t) \equiv 0$. 
    Then, there exists a quantum algorithm that prepares an $\epsilon$-approximation of the state $\ket{e^{-AT}u_0}$ with $\Omega(1)$ success probability and a flag indicating success, using queries to $O_L(s)$ and $O_H(s)$ a total number of times 
    \begin{equation}
        \widetilde{\mathcal{O}} \left( \Lambda_p^{1+1/p} \left(\frac{\|u_0\|}{\|u(T)\|}\right)^{2+1/p} \frac{T^{1+1/p}}{\epsilon^{1+1/p}}  \right)
    \end{equation}
    queries to $O_{\text{coef}}$ and $O_{\text{prep}}$ for $\mathcal{O}\left(\frac{\|u_0\|}{\|u(T)\|}\right)$ times, and $\mathcal{O}\left(\log\left(\frac{\|u_0\|\|L\|T}{\|u(T)\|\epsilon}\right)\right)$  ancilla qubits. 
\end{thm}

\begin{thm}\label{thm:ti_inhomo}
    Consider the ODE~\cref{eqn:inhom_general_diff_eq} with time-independent $A(t)\equiv A$. 
    Then, there exists a quantum algorithm that prepares an $\epsilon$-approximation of the state $\ket{u(T)}$ with $\Omega(1)$ success probability and a flag indicating success, using 
    \begin{enumerate}
        \item queries to $O_L(s)$ and $O_H(s)$ a total number of times 
        \begin{equation}
            \widetilde{\mathcal{O}}\left( \left( \frac{\|u_0\|+\|b\|_{L^1}}{\norm{u(T)} } \right)^{2+1/p} \Lambda_p^{1+1/p} \frac{  T^{1+1/p}}{\epsilon^{1+1/p}} \log^2(\Gamma_1 \|b\|_{C^2}) \right), 
        \end{equation} 
        \item queries to $O_{\text{prep}}$, $O_b$, $O_{\text{coef}}$ and  $O_{\text{coef}}'$ for $\mathcal{O}\left( \frac{ \|u_0\|+ \|b\|_{L^1}}{\norm{u(T)} } \right)$ times,  
        \item $\mathcal{O}(\log(\Gamma_1 \|b\|_{C^2} T/\epsilon))$ ancilla qubits, 
        \item $\mathcal{O}\left( \frac{ \|u_0\|+ \|b\|_{L^1}}{\norm{u(T)} } \right)$ additional one-qubit gates. 
    \end{enumerate}
\end{thm}

\section{Proof of \MakeLowercase{\texorpdfstring{\cref{thm:cap}}{}} }

Suppose we are given the oracle $O_{V_I}: \ket{\vr}\ket{0} \rightarrow \ket{\vr}\ket{V_I(\vr)}$. 
Here with a slight abuse of notation, $\ket{\vr}$ represents the binary encoding of some related index of $\vr$ after spatial discretization. 
Since $V_I$ is a diagonal matrix, $O_L(s) = e^{-i L s}$ can be constructed fast-forwardly with a single use of $O_{V_I}$ for any $s$~\cite{ahokas2004improved}. 
Then, the select oracle $\text{SEL}_L(s) = \sum_{j=0}^M \ket{j}\bra{j} \otimes e^{-iLk_j s}$ can be constructed by the same approach as in~\cref{app:proofs_oracles} using $\mathcal{O}(\log(M))$ queries to $O_{V_I}$. 

For the matrix $V_R$, we assume the sparse input oracle $O_{V_R}: \ket{\vr}\ket{s}\ket{0} \rightarrow \ket{\vr}\ket{s}\ket{V_R(\vr,s)}$, where, similar to $O_{V_I}$, $\ket{s}$ represents the encoding of the related discrete index (and may vary on different sub-intervals). 
Note that $V_R$ is a diagonal matrix, and $-\Delta_\vr/2$ after spatial discretization is a tri-diagonal matrix with diagonal entries $1/N^2$ and off-diagonal entries $-1/(2N^2)$. 
For a fixed time step $h$ and any integer $m$ such that $[mh,(m+1)h] \subset [0,T]$, we may construct a HAM-T oracle that block encodes $H(t) = -\frac12 \Delta_\vr + V_R(t)$ for $t\in[mh,(m+1)h]$, namely 
\begin{equation}
    \bra{0}_a \text{HAM-T}_{H,m} \ket{0}_a = \sum_{l=0}^{M_H-1} \ket{l}\bra{l} \otimes \frac{H(mh+lh/M_H)}{\alpha_{H}} 
\end{equation}
with $\mathcal{O}(1)$ uses of $O_{V_R}$~\cite{GilyenSuLowEtAl2019}. 
Here $\alpha_H = \mathcal{O}(N^2+\max_t\|V_R(t)\|)$, and $M_H$ is the number of the grid points used in each step of the truncated Dyson series method. 

Now we construct the HAM-T oracle of interaction picture Hamiltonian $H_I$. 
We first construct the select oracles
\begin{equation}
    \text{SEL}_{L,m} = \sum_{j=0}^M \sum_{l=0}^{M_H-1} \ket{j}\bra{j} \otimes \ket{l}\bra{l} \otimes e^{-i L k_j (mh+lh/M_H)}
\end{equation}
and 
\begin{equation}
    \text{SEL}_{L,m}' = \sum_{j=0}^M \sum_{l=0}^{M_H-1} \ket{j}\bra{j} \otimes \ket{l}\bra{l} \otimes e^{i L k_j (mh+lh/M_H)}
\end{equation}
using the same approach as in~\cref{app:proofs_oracles} with $\mathcal{O}(\log(M_H))$ queries to $\text{SEL}_L(s)$. 
Then 
\begin{equation}
    \text{HAM-T}_{H_I,m} \coloneqq (I_{n_a} \otimes \text{SEL}'_{L,m}) (I_{\log(M)} \otimes \text{HAM-T}_{H,m} ) (I_{n_a} \otimes \text{SEL}_{L,m} ) 
\end{equation}
gives the HAM-T oracle of $H_I(t;k)$ that 
\begin{equation}
    \bra{0}_a \text{HAM-T}_{H_I,m} \ket{0}_a = \sum_{j=0}^M \sum_{l=0}^{M_H-1} \ket{j}\bra{j}\otimes \ket{l}\bra{l} \otimes \frac{H_I (mh+lh/M_H;k_j)}{\alpha_{H}}. 
\end{equation}
Constructing $\text{HAM-T}_{H_I,m}$ requires $\mathcal{O}(\log(M_H)\log(M))$ queries to $O_{V_I}$ and $O_{V_R}$. 

Notice that this $\text{HAM-T}_{H_I,m}$ serves as the input model of the truncated Dyson series method in~\cite{LowWiebe2019} (which is denoted by $\text{HAM-T}_j$ there). 
By~\cite[Corollary 4]{LowWiebe2019}, for $\epsilon'>0$, we may implement a select oracle 
\begin{equation}
    \text{SEL}_{W} = \sum_{j=0}^M  \ket{j}\bra{j}\otimes W_j
\end{equation}
where 
\begin{equation}
    \norm{W_j - \left(\mc{T} e^{- i \int_0^T H_I(s;k_j) \ud s}\right)} \leq \epsilon', 
\end{equation}
with failure probability at most $\mathcal{O}(\epsilon')$. 
The number of the queries to $\text{HAM-T}_{H_I,m}$ is $\mathcal{O}(\alpha_H T \log(\alpha_H T/\epsilon'))$, and $M_H$ should be choose as $\mathcal{O}\left( \frac{T}{\alpha_H \epsilon} (\alpha_H^2 + K + \max_t\|V_R'(t)\|) \right)$. 
Then, the operator
\begin{equation}
    \text{SEL}_U \coloneqq \text{SEL}_L(T) \text{SEL}_{W} \text{SEL}_L(-T) = \sum_{j=0}^M  \ket{j}\bra{j}\otimes U_j
\end{equation}
where 
\begin{equation}
    \norm{U_j - e^{-iLk_j T}\left(\mc{T} e^{- i \int_0^T H_I(s;k_j) \ud s}\right) e^{iLk_j T}} \leq \epsilon'. 
\end{equation}

The operator $\text{SEL}_U$ serves as the select oracle in the LCU step. 
After the LCU as in our general algorithm, we obtain a quantum state $\frac{1}{\|c\|_1 \|u_0\|} \ket{0} \widetilde{u} + \ket{\perp}$, where 
\begin{equation}
    \widetilde{u} = \sum_{j=0}^M c_j U_j u_0. 
\end{equation}
The final error in the quantum state can be bounded as 
\begin{equation}
    \begin{split}
        \norm{\ket{u(T)} - \ket{\widetilde{u}}} & \leq \frac{2}{\|u(T)\|} \norm{u(T)-\widetilde{u}} \\
        & \leq \frac{2}{\|u(T)\|} \norm{u(T) - \sum_{j=0}^M c_j \mc{T} e^{- i \int_0^T (H(s)+k_j L) \ud s} u_0}  + \frac{2}{\norm{u(T)}}\sum_{j=0}^M |c_j| \norm{\mc{T} e^{- i \int_0^T (H(s)+k_j L) \ud s}-U_j} \norm{u_0} \\
        & \leq \frac{2\|u_0\|}{\|u(T)\|} \norm{\mc{T}e^{-\int_0^T A(s) \ud s} - \sum_{j=0}^M c_j \mc{T} e^{- i \int_0^T (H(s)+k_j L) \ud s}} + \frac{2}{\norm{u(T)}}\norm{c}_1 \norm{u_0} \epsilon', 
    \end{split}
\end{equation}
where the first part is the quadrature error, and the second part is the simulation error. 
To bound the overall error by $\epsilon$, we choose 
\begin{equation}
    M = \frac{\|V_I\|T}{\epsilon^2}, \quad \epsilon' = \frac{\epsilon \norm{u(T)}}{ 4\norm{c}_1 \norm{u_0}}. 
\end{equation}
With this choice and by $\norm{c}_1 = \mathcal{O}(1)$, $\alpha_H = \mathcal{O}(N^2+\max_t\|V_R(t)\|)$ and $M_H=\mathcal{O}\left( \frac{T}{\alpha_H \epsilon} (\alpha_H^2 + K + \max_t\|V_R'(t)\|) \right)$, the number of queries to $O_{V_I}$ and $O_{V_R}$ in each run of the LCU step becomes 
\begin{equation}
\begin{split}
    &\quad \mathcal{O}\left( \alpha_H T \log\left(\frac{\alpha_H T}{\epsilon'}\right) \log(M_H)\log(M) \right) \\
    &= \mathcal{O}\left( (N^2+\max_t\|V_R(t)\|) T \log\left(\frac{\norm{u_0}(N^2+\max_t\|V_R(t)\|)T}{\norm{u(T)}\epsilon} \right) \log\left(\frac{T(N+\max_t\|V_R'(t)\|)}{\epsilon}\right)\log\left(\frac{\|V_I\|T}{\epsilon}\right) \right). 
\end{split}
\end{equation}
With amplitude amplification, the number of repeats to get a success is $\mathcal{O}(\norm{c}_1\norm{u_0} / \norm{\widetilde{u}} ) = \mathcal{O}(\norm{u_0} / \norm{u(T)})$. 
This completes the proof.

%% file: main.bbl
\begin{thebibliography}{44}%
\makeatletter
\providecommand \@ifxundefined [1]{%
 \@ifx{#1\undefined}
}%
\providecommand \@ifnum [1]{%
 \ifnum #1\expandafter \@firstoftwo
 \else \expandafter \@secondoftwo
 \fi
}%
\providecommand \@ifx [1]{%
 \ifx #1\expandafter \@firstoftwo
 \else \expandafter \@secondoftwo
 \fi
}%
\providecommand \natexlab [1]{#1}%
\providecommand \enquote  [1]{``#1''}%
\providecommand \bibnamefont  [1]{#1}%
\providecommand \bibfnamefont [1]{#1}%
\providecommand \citenamefont [1]{#1}%
\providecommand \href@noop [0]{\@secondoftwo}%
\providecommand \href [0]{\begingroup \@sanitize@url \@href}%
\providecommand \@href[1]{\@@startlink{#1}\@@href}%
\providecommand \@@href[1]{\endgroup#1\@@endlink}%
\providecommand \@sanitize@url [0]{\catcode `\\12\catcode `\$12\catcode
  `\&12\catcode `\#12\catcode `\^12\catcode `\_12\catcode `\%12\relax}%
\providecommand \@@startlink[1]{}%
\providecommand \@@endlink[0]{}%
\providecommand \url  [0]{\begingroup\@sanitize@url \@url }%
\providecommand \@url [1]{\endgroup\@href {#1}{\urlprefix }}%
\providecommand \urlprefix  [0]{URL }%
\providecommand \Eprint [0]{\href }%
\providecommand \doibase [0]{https://doi.org/}%
\providecommand \selectlanguage [0]{\@gobble}%
\providecommand \bibinfo  [0]{\@secondoftwo}%
\providecommand \bibfield  [0]{\@secondoftwo}%
\providecommand \translation [1]{[#1]}%
\providecommand \BibitemOpen [0]{}%
\providecommand \bibitemStop [0]{}%
\providecommand \bibitemNoStop [0]{.\EOS\space}%
\providecommand \EOS [0]{\spacefactor3000\relax}%
\providecommand \BibitemShut  [1]{\csname bibitem#1\endcsname}%
\let\auto@bib@innerbib\@empty
\bibitem [{\citenamefont {Nielsen}\ and\ \citenamefont
  {Chuang}(2000)}]{NielsenChuang2000}%
  \BibitemOpen
  \bibfield  {author} {\bibinfo {author} {\bibfnamefont {M.~A.}\ \bibnamefont
  {Nielsen}}\ and\ \bibinfo {author} {\bibfnamefont {I.}~\bibnamefont
  {Chuang}},\ }\href@noop {} {\emph {\bibinfo {title} {Quantum computation and
  quantum information}}}\ (\bibinfo  {publisher} {Cambridge University Press},\
  \bibinfo {year} {2000})\BibitemShut {NoStop}%
\bibitem [{\citenamefont {Kitaev}\ \emph {et~al.}(2002)\citenamefont {Kitaev},
  \citenamefont {Shen},\ and\ \citenamefont {Vyalyi}}]{KitaevShenVyalyi2002}%
  \BibitemOpen
  \bibfield  {author} {\bibinfo {author} {\bibfnamefont {A.~Y.}\ \bibnamefont
  {Kitaev}}, \bibinfo {author} {\bibfnamefont {A.}~\bibnamefont {Shen}},\ and\
  \bibinfo {author} {\bibfnamefont {M.~N.}\ \bibnamefont {Vyalyi}},\
  }\href@noop {} {\emph {\bibinfo {title} {Classical and quantum
  computation}}},\ \bibinfo {number} {47}\ (\bibinfo  {publisher} {American
  Mathematical Soc.},\ \bibinfo {year} {2002})\BibitemShut {NoStop}%
\bibitem [{\citenamefont {Brassard}\ \emph {et~al.}(2002)\citenamefont
  {Brassard}, \citenamefont {Hoyer}, \citenamefont {Mosca},\ and\ \citenamefont
  {Tapp}}]{BrassardHoyerMoscaEtAl2002}%
  \BibitemOpen
  \bibfield  {author} {\bibinfo {author} {\bibfnamefont {G.}~\bibnamefont
  {Brassard}}, \bibinfo {author} {\bibfnamefont {P.}~\bibnamefont {Hoyer}},
  \bibinfo {author} {\bibfnamefont {M.}~\bibnamefont {Mosca}},\ and\ \bibinfo
  {author} {\bibfnamefont {A.}~\bibnamefont {Tapp}},\ }\href
  {https://doi.org/10.1090/conm/305/05215} {\bibfield  {journal} {\bibinfo
  {journal} {Contemp. Math.}\ }\textbf {\bibinfo {volume} {305}},\ \bibinfo
  {pages} {53} (\bibinfo {year} {2002})}\BibitemShut {NoStop}%
\bibitem [{\citenamefont {Harrow}\ \emph {et~al.}(2009)\citenamefont {Harrow},
  \citenamefont {Hassidim},\ and\ \citenamefont
  {Lloyd}}]{HarrowHassidimLloyd2009}%
  \BibitemOpen
  \bibfield  {author} {\bibinfo {author} {\bibfnamefont {A.~W.}\ \bibnamefont
  {Harrow}}, \bibinfo {author} {\bibfnamefont {A.}~\bibnamefont {Hassidim}},\
  and\ \bibinfo {author} {\bibfnamefont {S.}~\bibnamefont {Lloyd}},\ }\href
  {https://doi.org/10.1103/PhysRevLett.103.150502} {\bibfield  {journal}
  {\bibinfo  {journal} {Phys. Rev. Lett.}\ }\textbf {\bibinfo {volume} {103}},\
  \bibinfo {pages} {150502} (\bibinfo {year} {2009})}\BibitemShut {NoStop}%
\bibitem [{\citenamefont {Berry}(2014)}]{Berry2014}%
  \BibitemOpen
  \bibfield  {author} {\bibinfo {author} {\bibfnamefont {D.~W.}\ \bibnamefont
  {Berry}},\ }\href {https://doi.org/10.1088/1751-8113/47/10/105301} {\bibfield
   {journal} {\bibinfo  {journal} {Journal of Physics A: Mathematical and
  Theoretical}\ }\textbf {\bibinfo {volume} {47}},\ \bibinfo {pages} {105301}
  (\bibinfo {year} {2014})}\BibitemShut {NoStop}%
\bibitem [{\citenamefont {Childs}\ and\ \citenamefont
  {Liu}(2020)}]{ChildsLiu2020}%
  \BibitemOpen
  \bibfield  {author} {\bibinfo {author} {\bibfnamefont {A.~M.}\ \bibnamefont
  {Childs}}\ and\ \bibinfo {author} {\bibfnamefont {J.-P.}\ \bibnamefont
  {Liu}},\ }\href {https://doi.org/10.1007/s00220-020-03699-z} {\bibfield
  {journal} {\bibinfo  {journal} {Communications in Mathematical Physics}\
  }\textbf {\bibinfo {volume} {375}},\ \bibinfo {pages} {1427} (\bibinfo {year}
  {2020})}\BibitemShut {NoStop}%
\bibitem [{\citenamefont {Childs}\ and\ \citenamefont
  {Wiebe}(2012)}]{ChildsWiebe2012}%
  \BibitemOpen
  \bibfield  {author} {\bibinfo {author} {\bibfnamefont {A.~M.}\ \bibnamefont
  {Childs}}\ and\ \bibinfo {author} {\bibfnamefont {N.}~\bibnamefont {Wiebe}},\
  }\href@noop {} {\bibinfo {title} {Hamiltonian simulation using linear
  combinations of unitary operations}} (\bibinfo {year} {2012}),\ \Eprint
  {https://arxiv.org/abs/1202.5822} {arXiv:1202.5822} \BibitemShut {NoStop}%
\bibitem [{\citenamefont {Berry}\ \emph
  {et~al.}(2015{\natexlab{a}})\citenamefont {Berry}, \citenamefont {Childs},
  \citenamefont {Cleve}, \citenamefont {Kothari},\ and\ \citenamefont
  {Somma}}]{BerryChildsCleveEtAl2015}%
  \BibitemOpen
  \bibfield  {author} {\bibinfo {author} {\bibfnamefont {D.~W.}\ \bibnamefont
  {Berry}}, \bibinfo {author} {\bibfnamefont {A.~M.}\ \bibnamefont {Childs}},
  \bibinfo {author} {\bibfnamefont {R.}~\bibnamefont {Cleve}}, \bibinfo
  {author} {\bibfnamefont {R.}~\bibnamefont {Kothari}},\ and\ \bibinfo {author}
  {\bibfnamefont {R.~D.}\ \bibnamefont {Somma}},\ }\href
  {https://doi.org/10.1103/PhysRevLett.114.090502} {\bibfield  {journal}
  {\bibinfo  {journal} {Phys. Rev. Lett.}\ }\textbf {\bibinfo {volume} {114}},\
  \bibinfo {pages} {090502} (\bibinfo {year} {2015}{\natexlab{a}})}\BibitemShut
  {NoStop}%
\bibitem [{\citenamefont {Berry}\ \emph
  {et~al.}(2015{\natexlab{b}})\citenamefont {Berry}, \citenamefont {Childs},\
  and\ \citenamefont {Kothari}}]{BerryChildsKothari2015}%
  \BibitemOpen
  \bibfield  {author} {\bibinfo {author} {\bibfnamefont {D.~W.}\ \bibnamefont
  {Berry}}, \bibinfo {author} {\bibfnamefont {A.~M.}\ \bibnamefont {Childs}},\
  and\ \bibinfo {author} {\bibfnamefont {R.}~\bibnamefont {Kothari}},\ }\href
  {https://doi.org/10.1109/FOCS.2015.54} {\bibfield  {journal} {\bibinfo
  {journal} {Proceedings of the 56th IEEE Symposium on Foundations of Computer
  Science}\ ,\ \bibinfo {pages} {792}} (\bibinfo {year}
  {2015}{\natexlab{b}})}\BibitemShut {NoStop}%
\bibitem [{\citenamefont {Childs}\ \emph {et~al.}(2017)\citenamefont {Childs},
  \citenamefont {Kothari},\ and\ \citenamefont
  {Somma}}]{ChildsKothariSomma2017}%
  \BibitemOpen
  \bibfield  {author} {\bibinfo {author} {\bibfnamefont {A.~M.}\ \bibnamefont
  {Childs}}, \bibinfo {author} {\bibfnamefont {R.}~\bibnamefont {Kothari}},\
  and\ \bibinfo {author} {\bibfnamefont {R.~D.}\ \bibnamefont {Somma}},\ }\href
  {https://doi.org/10.1137/16M1087072} {\bibfield  {journal} {\bibinfo
  {journal} {SIAM J. Comput.}\ }\textbf {\bibinfo {volume} {46}},\ \bibinfo
  {pages} {1920} (\bibinfo {year} {2017})}\BibitemShut {NoStop}%
\bibitem [{\citenamefont {Low}\ and\ \citenamefont
  {Chuang}(2017)}]{LowChuang2017}%
  \BibitemOpen
  \bibfield  {author} {\bibinfo {author} {\bibfnamefont {G.~H.}\ \bibnamefont
  {Low}}\ and\ \bibinfo {author} {\bibfnamefont {I.~L.}\ \bibnamefont
  {Chuang}},\ }\href {https://doi.org/10.1103/PhysRevLett.118.010501}
  {\bibfield  {journal} {\bibinfo  {journal} {Phys. Rev. Lett.}\ }\textbf
  {\bibinfo {volume} {118}},\ \bibinfo {pages} {010501} (\bibinfo {year}
  {2017})}\BibitemShut {NoStop}%
\bibitem [{\citenamefont {Gily{\'e}n}\ \emph {et~al.}(2019)\citenamefont
  {Gily{\'e}n}, \citenamefont {Su}, \citenamefont {Low},\ and\ \citenamefont
  {Wiebe}}]{GilyenSuLowEtAl2019}%
  \BibitemOpen
  \bibfield  {author} {\bibinfo {author} {\bibfnamefont {A.}~\bibnamefont
  {Gily{\'e}n}}, \bibinfo {author} {\bibfnamefont {Y.}~\bibnamefont {Su}},
  \bibinfo {author} {\bibfnamefont {G.~H.}\ \bibnamefont {Low}},\ and\ \bibinfo
  {author} {\bibfnamefont {N.}~\bibnamefont {Wiebe}},\ }in\ \href
  {https://doi.org/10.1145/3313276.3316366} {\emph {\bibinfo {booktitle}
  {Proceedings of the 51st Annual ACM SIGACT Symposium on Theory of
  Computing}}}\ (\bibinfo {year} {2019})\ pp.\ \bibinfo {pages}
  {193--204}\BibitemShut {NoStop}%
\bibitem [{\citenamefont {Dong}\ \emph {et~al.}(2022)\citenamefont {Dong},
  \citenamefont {Lin},\ and\ \citenamefont {Tong}}]{DongLinTong2022}%
  \BibitemOpen
  \bibfield  {author} {\bibinfo {author} {\bibfnamefont {Y.}~\bibnamefont
  {Dong}}, \bibinfo {author} {\bibfnamefont {L.}~\bibnamefont {Lin}},\ and\
  \bibinfo {author} {\bibfnamefont {Y.}~\bibnamefont {Tong}},\ }\href
  {https://doi.org/10.1103/prxquantum.3.040305} {\bibfield  {journal} {\bibinfo
   {journal} {PRX Quantum}\ }\textbf {\bibinfo {volume} {3}},\ \bibinfo {pages}
  {040305} (\bibinfo {year} {2022})}\BibitemShut {NoStop}%
\bibitem [{\citenamefont {Mahapatra}\ and\ \citenamefont
  {Sathyamurthy}(1997)}]{MahapatraSathyamurthy1997}%
  \BibitemOpen
  \bibfield  {author} {\bibinfo {author} {\bibfnamefont {S.}~\bibnamefont
  {Mahapatra}}\ and\ \bibinfo {author} {\bibfnamefont {N.}~\bibnamefont
  {Sathyamurthy}},\ }\href {https://doi.org/10.1039/A605778K} {\bibfield
  {journal} {\bibinfo  {journal} {J. Chem. Soc., Faraday Trans.,}\ }\textbf
  {\bibinfo {volume} {93}},\ \bibinfo {pages} {773} (\bibinfo {year}
  {1997})}\BibitemShut {NoStop}%
\bibitem [{\citenamefont {Vibok}\ and\ \citenamefont
  {Balint-Kurti}(1992)}]{VibokBalint-Kurti1992}%
  \BibitemOpen
  \bibfield  {author} {\bibinfo {author} {\bibfnamefont {A.}~\bibnamefont
  {Vibok}}\ and\ \bibinfo {author} {\bibfnamefont {G.}~\bibnamefont
  {Balint-Kurti}},\ }\href {https://doi.org/10.1021/j100201a012} {\bibfield
  {journal} {\bibinfo  {journal} {J. Phys. Chem.}\ }\textbf {\bibinfo {volume}
  {96}},\ \bibinfo {pages} {8712} (\bibinfo {year} {1992})}\BibitemShut
  {NoStop}%
\bibitem [{\citenamefont {Datta}(2005)}]{Datta2005}%
  \BibitemOpen
  \bibfield  {author} {\bibinfo {author} {\bibfnamefont {S.}~\bibnamefont
  {Datta}},\ }\href@noop {} {\emph {\bibinfo {title} {Quantum transport: atom
  to transistor}}}\ (\bibinfo  {publisher} {Cambridge Univ. Pr.},\ \bibinfo
  {year} {2005})\BibitemShut {NoStop}%
\bibitem [{\citenamefont {Child}(1991)}]{Child1991}%
  \BibitemOpen
  \bibfield  {author} {\bibinfo {author} {\bibfnamefont {M.}~\bibnamefont
  {Child}},\ }\href {https://doi.org/10.1080/00268979100100041} {\bibfield
  {journal} {\bibinfo  {journal} {Mol. Phys.}\ }\textbf {\bibinfo {volume}
  {72}},\ \bibinfo {pages} {89} (\bibinfo {year} {1991})}\BibitemShut {NoStop}%
\bibitem [{\citenamefont {Muga}\ \emph {et~al.}(2004)\citenamefont {Muga},
  \citenamefont {Palao}, \citenamefont {Navarro},\ and\ \citenamefont
  {Egusquiza}}]{MugaPalaoNavarroEtAl2004}%
  \BibitemOpen
  \bibfield  {author} {\bibinfo {author} {\bibfnamefont {J.}~\bibnamefont
  {Muga}}, \bibinfo {author} {\bibfnamefont {J.}~\bibnamefont {Palao}},
  \bibinfo {author} {\bibfnamefont {B.}~\bibnamefont {Navarro}},\ and\ \bibinfo
  {author} {\bibfnamefont {I.}~\bibnamefont {Egusquiza}},\ }\href
  {https://doi.org/10.1016/j.physrep.2004.03.002} {\bibfield  {journal}
  {\bibinfo  {journal} {Phys. Rep.}\ }\textbf {\bibinfo {volume} {395}},\
  \bibinfo {pages} {357} (\bibinfo {year} {2004})}\BibitemShut {NoStop}%
\bibitem [{\citenamefont {Zeng}\ \emph {et~al.}(2022)\citenamefont {Zeng},
  \citenamefont {Sun},\ and\ \citenamefont {Yuan}}]{ZengSunYuan2022}%
  \BibitemOpen
  \bibfield  {author} {\bibinfo {author} {\bibfnamefont {P.}~\bibnamefont
  {Zeng}}, \bibinfo {author} {\bibfnamefont {J.}~\bibnamefont {Sun}},\ and\
  \bibinfo {author} {\bibfnamefont {X.}~\bibnamefont {Yuan}},\ }\href@noop {}
  {\bibinfo {title} {Universal quantum algorithmic cooling on a quantum
  computer}} (\bibinfo {year} {2022}),\ \Eprint
  {https://arxiv.org/abs/2109.15304} {arXiv:2109.15304 [quant-ph]} \BibitemShut
  {NoStop}%
\bibitem [{\citenamefont {Huo}\ and\ \citenamefont {Li}(2023)}]{HuoLi2023}%
  \BibitemOpen
  \bibfield  {author} {\bibinfo {author} {\bibfnamefont {M.}~\bibnamefont
  {Huo}}\ and\ \bibinfo {author} {\bibfnamefont {Y.}~\bibnamefont {Li}},\
  }\href {https://doi.org/10.22331/q-2023-02-09-916} {\bibfield  {journal}
  {\bibinfo  {journal} {Quantum}\ }\textbf {\bibinfo {volume} {7}},\ \bibinfo
  {pages} {916} (\bibinfo {year} {2023})}\BibitemShut {NoStop}%
\bibitem [{\citenamefont {Tong}\ \emph {et~al.}(2021)\citenamefont {Tong},
  \citenamefont {An}, \citenamefont {Wiebe},\ and\ \citenamefont
  {Lin}}]{TongAnWiebe2021}%
  \BibitemOpen
  \bibfield  {author} {\bibinfo {author} {\bibfnamefont {Y.}~\bibnamefont
  {Tong}}, \bibinfo {author} {\bibfnamefont {D.}~\bibnamefont {An}}, \bibinfo
  {author} {\bibfnamefont {N.}~\bibnamefont {Wiebe}},\ and\ \bibinfo {author}
  {\bibfnamefont {L.}~\bibnamefont {Lin}},\ }\bibfield  {journal} {\bibinfo
  {journal} {Phys. Rev. A}\ }\textbf {\bibinfo {volume} {104}},\ \href
  {https://doi.org/10.1103/physreva.104.032422} {10.1103/physreva.104.032422}
  (\bibinfo {year} {2021})\BibitemShut {NoStop}%
\bibitem [{\citenamefont {Wiebe}\ \emph {et~al.}(2010)\citenamefont {Wiebe},
  \citenamefont {Berry}, \citenamefont {H{\o}yer},\ and\ \citenamefont
  {Sanders}}]{WiebeBerryHoyerEtAl2010}%
  \BibitemOpen
  \bibfield  {author} {\bibinfo {author} {\bibfnamefont {N.}~\bibnamefont
  {Wiebe}}, \bibinfo {author} {\bibfnamefont {D.}~\bibnamefont {Berry}},
  \bibinfo {author} {\bibfnamefont {P.}~\bibnamefont {H{\o}yer}},\ and\
  \bibinfo {author} {\bibfnamefont {B.~C.}\ \bibnamefont {Sanders}},\ }\href
  {https://doi.org/10.1088/1751-8113/43/6/065203} {\bibfield  {journal}
  {\bibinfo  {journal} {Journal of Physics A: Mathematical and Theoretical}\
  }\textbf {\bibinfo {volume} {43}},\ \bibinfo {pages} {065203} (\bibinfo
  {year} {2010})}\BibitemShut {NoStop}%
\bibitem [{\citenamefont {Low}\ and\ \citenamefont
  {Wiebe}(2019)}]{LowWiebe2019}%
  \BibitemOpen
  \bibfield  {author} {\bibinfo {author} {\bibfnamefont {G.~H.}\ \bibnamefont
  {Low}}\ and\ \bibinfo {author} {\bibfnamefont {N.}~\bibnamefont {Wiebe}},\
  }\href@noop {} {\bibinfo {title} {{H}amiltonian simulation in the interaction
  picture}} (\bibinfo {year} {2019}),\ \Eprint
  {https://arxiv.org/abs/1805.00675} {arXiv:1805.00675} \BibitemShut {NoStop}%
\bibitem [{\citenamefont {An}\ \emph {et~al.}(2022{\natexlab{a}})\citenamefont
  {An}, \citenamefont {Fang},\ and\ \citenamefont {Lin}}]{AnFangLin2022}%
  \BibitemOpen
  \bibfield  {author} {\bibinfo {author} {\bibfnamefont {D.}~\bibnamefont
  {An}}, \bibinfo {author} {\bibfnamefont {D.}~\bibnamefont {Fang}},\ and\
  \bibinfo {author} {\bibfnamefont {L.}~\bibnamefont {Lin}},\ }\href
  {https://doi.org/10.22331/q-2022-04-15-690} {\bibfield  {journal} {\bibinfo
  {journal} {Quantum}\ }\textbf {\bibinfo {volume} {6}},\ \bibinfo {pages}
  {690} (\bibinfo {year} {2022}{\natexlab{a}})}\BibitemShut {NoStop}%
\bibitem [{\citenamefont {Childs}\ \emph {et~al.}(2021)\citenamefont {Childs},
  \citenamefont {Su}, \citenamefont {Tran}, \citenamefont {Wiebe},\ and\
  \citenamefont {Zhu}}]{ChildsSuTranEtAl2020}%
  \BibitemOpen
  \bibfield  {author} {\bibinfo {author} {\bibfnamefont {A.~M.}\ \bibnamefont
  {Childs}}, \bibinfo {author} {\bibfnamefont {Y.}~\bibnamefont {Su}}, \bibinfo
  {author} {\bibfnamefont {M.~C.}\ \bibnamefont {Tran}}, \bibinfo {author}
  {\bibfnamefont {N.}~\bibnamefont {Wiebe}},\ and\ \bibinfo {author}
  {\bibfnamefont {S.}~\bibnamefont {Zhu}},\ }\href
  {https://doi.org/10.1103/PhysRevX.11.011020} {\bibfield  {journal} {\bibinfo
  {journal} {Phys. Rev. X}\ }\textbf {\bibinfo {volume} {11}},\ \bibinfo
  {pages} {011020} (\bibinfo {year} {2021})}\BibitemShut {NoStop}%
\bibitem [{\citenamefont {Berry}\ \emph {et~al.}(2017)\citenamefont {Berry},
  \citenamefont {Childs}, \citenamefont {Ostrander},\ and\ \citenamefont
  {Wang}}]{BerryChildsOstranderEtAl2017}%
  \BibitemOpen
  \bibfield  {author} {\bibinfo {author} {\bibfnamefont {D.~W.}\ \bibnamefont
  {Berry}}, \bibinfo {author} {\bibfnamefont {A.~M.}\ \bibnamefont {Childs}},
  \bibinfo {author} {\bibfnamefont {A.}~\bibnamefont {Ostrander}},\ and\
  \bibinfo {author} {\bibfnamefont {G.}~\bibnamefont {Wang}},\ }\href
  {https://doi.org/10.1007/s00220-017-3002-y} {\bibfield  {journal} {\bibinfo
  {journal} {Communications in Mathematical Physics}\ }\textbf {\bibinfo
  {volume} {356}},\ \bibinfo {pages} {1057} (\bibinfo {year}
  {2017})}\BibitemShut {NoStop}%
\bibitem [{\citenamefont {Liu}\ \emph {et~al.}(2021)\citenamefont {Liu},
  \citenamefont {Kolden}, \citenamefont {Krovi}, \citenamefont {Loureiro},
  \citenamefont {Trivisa},\ and\ \citenamefont
  {Childs}}]{LiuKoldenKroviEtAl2021}%
  \BibitemOpen
  \bibfield  {author} {\bibinfo {author} {\bibfnamefont {J.-P.}\ \bibnamefont
  {Liu}}, \bibinfo {author} {\bibfnamefont {H.~{\O}.}\ \bibnamefont {Kolden}},
  \bibinfo {author} {\bibfnamefont {H.~K.}\ \bibnamefont {Krovi}}, \bibinfo
  {author} {\bibfnamefont {N.~F.}\ \bibnamefont {Loureiro}}, \bibinfo {author}
  {\bibfnamefont {K.}~\bibnamefont {Trivisa}},\ and\ \bibinfo {author}
  {\bibfnamefont {A.~M.}\ \bibnamefont {Childs}},\ }\bibfield  {journal}
  {\bibinfo  {journal} {Proceedings of the National Academy of Sciences}\
  }\textbf {\bibinfo {volume} {118}},\ \href
  {https://doi.org/10.1073/pnas.2026805118} {10.1073/pnas.2026805118} (\bibinfo
  {year} {2021})\BibitemShut {NoStop}%
\bibitem [{\citenamefont {An}\ \emph {et~al.}(2022{\natexlab{b}})\citenamefont
  {An}, \citenamefont {Liu}, \citenamefont {Wang},\ and\ \citenamefont
  {Zhao}}]{an2022theory}%
  \BibitemOpen
  \bibfield  {author} {\bibinfo {author} {\bibfnamefont {D.}~\bibnamefont
  {An}}, \bibinfo {author} {\bibfnamefont {J.-P.}\ \bibnamefont {Liu}},
  \bibinfo {author} {\bibfnamefont {D.}~\bibnamefont {Wang}},\ and\ \bibinfo
  {author} {\bibfnamefont {Q.}~\bibnamefont {Zhao}},\ }\href@noop {} {\bibinfo
  {title} {A theory of quantum differential equation solvers: limitations and
  fast-forwarding}} (\bibinfo {year} {2022}{\natexlab{b}}),\ \Eprint
  {https://arxiv.org/abs/2211.05246} {arXiv:2211.05246} \BibitemShut {NoStop}%
\bibitem [{\citenamefont {Krovi}(2023)}]{Krovi2022}%
  \BibitemOpen
  \bibfield  {author} {\bibinfo {author} {\bibfnamefont {H.}~\bibnamefont
  {Krovi}},\ }\href {https://doi.org/10.22331/q-2023-02-02-913} {\bibfield
  {journal} {\bibinfo  {journal} {Quantum}\ }\textbf {\bibinfo {volume} {7}},\
  \bibinfo {pages} {913} (\bibinfo {year} {2023})}\BibitemShut {NoStop}%
\bibitem [{\citenamefont {Berry}\ and\ \citenamefont
  {Costa}(2022)}]{berry2022quantum}%
  \BibitemOpen
  \bibfield  {author} {\bibinfo {author} {\bibfnamefont {D.~W.}\ \bibnamefont
  {Berry}}\ and\ \bibinfo {author} {\bibfnamefont {P.}~\bibnamefont {Costa}},\
  }\href@noop {} {\bibinfo {title} {Quantum algorithm for time-dependent
  differential equations using dyson series}} (\bibinfo {year} {2022}),\
  \Eprint {https://arxiv.org/abs/2212.03544} {arXiv:2212.03544} \BibitemShut
  {NoStop}%
\bibitem [{\citenamefont {Costa}\ \emph {et~al.}(2022)\citenamefont {Costa},
  \citenamefont {An}, \citenamefont {Sanders}, \citenamefont {Su},
  \citenamefont {Babbush},\ and\ \citenamefont {Berry}}]{CostaAnYuvalEtAl2022}%
  \BibitemOpen
  \bibfield  {author} {\bibinfo {author} {\bibfnamefont {P.~C.}\ \bibnamefont
  {Costa}}, \bibinfo {author} {\bibfnamefont {D.}~\bibnamefont {An}}, \bibinfo
  {author} {\bibfnamefont {Y.~R.}\ \bibnamefont {Sanders}}, \bibinfo {author}
  {\bibfnamefont {Y.}~\bibnamefont {Su}}, \bibinfo {author} {\bibfnamefont
  {R.}~\bibnamefont {Babbush}},\ and\ \bibinfo {author} {\bibfnamefont {D.~W.}\
  \bibnamefont {Berry}},\ }\href {https://doi.org/10.1103/PRXQuantum.3.040303}
  {\bibfield  {journal} {\bibinfo  {journal} {PRX Quantum}\ }\textbf {\bibinfo
  {volume} {3}},\ \bibinfo {pages} {040303} (\bibinfo {year}
  {2022})}\BibitemShut {NoStop}%
\bibitem [{\citenamefont {Fang}\ \emph {et~al.}(2023)\citenamefont {Fang},
  \citenamefont {Lin},\ and\ \citenamefont {Tong}}]{FangLinTong2022}%
  \BibitemOpen
  \bibfield  {author} {\bibinfo {author} {\bibfnamefont {D.}~\bibnamefont
  {Fang}}, \bibinfo {author} {\bibfnamefont {L.}~\bibnamefont {Lin}},\ and\
  \bibinfo {author} {\bibfnamefont {Y.}~\bibnamefont {Tong}},\ }\href
  {https://doi.org/10.22331/q-2023-03-20-955} {\bibfield  {journal} {\bibinfo
  {journal} {Quantum}\ }\textbf {\bibinfo {volume} {7}},\ \bibinfo {pages}
  {955} (\bibinfo {year} {2023})}\BibitemShut {NoStop}%
\bibitem [{\citenamefont {Berenger}(1994)}]{Berenger1994}%
  \BibitemOpen
  \bibfield  {author} {\bibinfo {author} {\bibfnamefont {J.-P.}\ \bibnamefont
  {Berenger}},\ }\href {https://doi.org/10.1006/jcph.1994.1159} {\bibfield
  {journal} {\bibinfo  {journal} {J. Comput. Phys.}\ }\textbf {\bibinfo
  {volume} {114}},\ \bibinfo {pages} {185} (\bibinfo {year}
  {1994})}\BibitemShut {NoStop}%
\bibitem [{\citenamefont {Zheng}(2007)}]{Zheng2007}%
  \BibitemOpen
  \bibfield  {author} {\bibinfo {author} {\bibfnamefont {C.}~\bibnamefont
  {Zheng}},\ }\href {https://doi.org/10.1016/j.jcp.2007.08.004} {\bibfield
  {journal} {\bibinfo  {journal} {J. Comput. Phys.}\ }\textbf {\bibinfo
  {volume} {227}},\ \bibinfo {pages} {537} (\bibinfo {year}
  {2007})}\BibitemShut {NoStop}%
\bibitem [{\citenamefont {Lindblad}(1976)}]{Lindblad1976}%
  \BibitemOpen
  \bibfield  {author} {\bibinfo {author} {\bibfnamefont {G.}~\bibnamefont
  {Lindblad}},\ }\href@noop {} {\bibfield  {journal} {\bibinfo  {journal}
  {Communications in Mathematical Physics}\ }\textbf {\bibinfo {volume} {48}},\
  \bibinfo {pages} {119} (\bibinfo {year} {1976})}\BibitemShut {NoStop}%
\bibitem [{\citenamefont {Gorini}\ \emph {et~al.}(1976)\citenamefont {Gorini},
  \citenamefont {Kossakowski},\ and\ \citenamefont
  {Sudarshan}}]{GoriniKossakowskiSudarshan1976}%
  \BibitemOpen
  \bibfield  {author} {\bibinfo {author} {\bibfnamefont {V.}~\bibnamefont
  {Gorini}}, \bibinfo {author} {\bibfnamefont {A.}~\bibnamefont
  {Kossakowski}},\ and\ \bibinfo {author} {\bibfnamefont {E.~C.~G.}\
  \bibnamefont {Sudarshan}},\ }\href@noop {} {\bibfield  {journal} {\bibinfo
  {journal} {J. Math. Phys.}\ }\textbf {\bibinfo {volume} {17}},\ \bibinfo
  {pages} {821} (\bibinfo {year} {1976})}\BibitemShut {NoStop}%
\bibitem [{\citenamefont {Landi}\ \emph {et~al.}(2022)\citenamefont {Landi},
  \citenamefont {Poletti},\ and\ \citenamefont
  {Schaller}}]{LandiPolettiSchaller2022}%
  \BibitemOpen
  \bibfield  {author} {\bibinfo {author} {\bibfnamefont {G.~T.}\ \bibnamefont
  {Landi}}, \bibinfo {author} {\bibfnamefont {D.}~\bibnamefont {Poletti}},\
  and\ \bibinfo {author} {\bibfnamefont {G.}~\bibnamefont {Schaller}},\
  }\href@noop {} {\bibfield  {journal} {\bibinfo  {journal} {Rev. Mod. Phys.}\
  }\textbf {\bibinfo {volume} {94}},\ \bibinfo {pages} {045006} (\bibinfo
  {year} {2022})}\BibitemShut {NoStop}%
\bibitem [{\citenamefont {Dalibard}\ \emph {et~al.}(1992)\citenamefont
  {Dalibard}, \citenamefont {Castin},\ and\ \citenamefont
  {M{\o}lmer}}]{DalibardCastinMolmer1992}%
  \BibitemOpen
  \bibfield  {author} {\bibinfo {author} {\bibfnamefont {J.}~\bibnamefont
  {Dalibard}}, \bibinfo {author} {\bibfnamefont {Y.}~\bibnamefont {Castin}},\
  and\ \bibinfo {author} {\bibfnamefont {K.}~\bibnamefont {M{\o}lmer}},\
  }\href@noop {} {\bibfield  {journal} {\bibinfo  {journal} {Phys. Rev. Lett.}\
  }\textbf {\bibinfo {volume} {68}},\ \bibinfo {pages} {580} (\bibinfo {year}
  {1992})}\BibitemShut {NoStop}%
\bibitem [{\citenamefont {Dum}\ \emph {et~al.}(1992)\citenamefont {Dum},
  \citenamefont {Zoller},\ and\ \citenamefont {Ritsch}}]{DumZollerRitsch1992}%
  \BibitemOpen
  \bibfield  {author} {\bibinfo {author} {\bibfnamefont {R.}~\bibnamefont
  {Dum}}, \bibinfo {author} {\bibfnamefont {P.}~\bibnamefont {Zoller}},\ and\
  \bibinfo {author} {\bibfnamefont {H.}~\bibnamefont {Ritsch}},\ }\href@noop {}
  {\bibfield  {journal} {\bibinfo  {journal} {Phys. Rev. A}\ }\textbf {\bibinfo
  {volume} {45}},\ \bibinfo {pages} {4879} (\bibinfo {year}
  {1992})}\BibitemShut {NoStop}%
\bibitem [{\citenamefont {Ahokas}(2004)}]{ahokas2004improved}%
  \BibitemOpen
  \bibfield  {author} {\bibinfo {author} {\bibfnamefont {G.~R.}\ \bibnamefont
  {Ahokas}},\ }\href {https://doi.org/10.11575/PRISM/22839} {\emph {\bibinfo
  {title} {Improved algorithms for approximate quantum {F}ourier transforms and
  sparse {H}amiltonian simulations}}}\ (\bibinfo  {publisher} {University of
  Calgary},\ \bibinfo {year} {2004})\BibitemShut {NoStop}%
\bibitem [{\citenamefont {Jin}\ \emph {et~al.}(2022)\citenamefont {Jin},
  \citenamefont {Liu},\ and\ \citenamefont {Yu}}]{JinLiuYu2022}%
  \BibitemOpen
  \bibfield  {author} {\bibinfo {author} {\bibfnamefont {S.}~\bibnamefont
  {Jin}}, \bibinfo {author} {\bibfnamefont {N.}~\bibnamefont {Liu}},\ and\
  \bibinfo {author} {\bibfnamefont {Y.}~\bibnamefont {Yu}},\ }\href@noop {}
  {\bibinfo {title} {Quantum simulation of partial differential equations via
  {S}chrodingerisation}} (\bibinfo {year} {2022}),\ \Eprint
  {https://arxiv.org/abs/2212.13969} {arXiv:2212.13969} \BibitemShut {NoStop}%
\bibitem [{\citenamefont {Kothari}(2014)}]{Kothari2014}%
  \BibitemOpen
  \bibfield  {author} {\bibinfo {author} {\bibfnamefont {R.}~\bibnamefont
  {Kothari}},\ }\href@noop {} {\bibinfo {title} {Efficient algorithms in
  quantum query complexity}} (\bibinfo {year} {2014})\BibitemShut {NoStop}%
\bibitem [{\citenamefont {Wang}\ \emph {et~al.}(2023)\citenamefont {Wang},
  \citenamefont {McArdle},\ and\ \citenamefont {Berta}}]{WangMcArdleBerta2023}%
  \BibitemOpen
  \bibfield  {author} {\bibinfo {author} {\bibfnamefont {S.}~\bibnamefont
  {Wang}}, \bibinfo {author} {\bibfnamefont {S.}~\bibnamefont {McArdle}},\ and\
  \bibinfo {author} {\bibfnamefont {M.}~\bibnamefont {Berta}},\ }\href@noop {}
  {\bibinfo {title} {Qubit-efficient randomized quantum algorithms for linear
  algebra}} (\bibinfo {year} {2023}),\ \Eprint
  {https://arxiv.org/abs/2302.01873} {arXiv:2302.01873 [quant-ph]} \BibitemShut
  {NoStop}%
\bibitem [{\citenamefont {Chakraborty}(2023)}]{Chakraborty2023}%
  \BibitemOpen
  \bibfield  {author} {\bibinfo {author} {\bibfnamefont {S.}~\bibnamefont
  {Chakraborty}},\ }\href@noop {} {\bibinfo {title} {Implementing linear
  combination of unitaries on intermediate-term quantum computers}} (\bibinfo
  {year} {2023}),\ \Eprint {https://arxiv.org/abs/2302.13555} {arXiv:2302.13555
  [quant-ph]} \BibitemShut {NoStop}%
\end{thebibliography}%
